\documentclass[a4paper, pdftex,pra,aps,twocolumn,twoside,nofootinbib]{revtex4-2}
\pdfoutput=1

\usepackage[utf8]{inputenc}
\usepackage[english]{babel}
\usepackage{amsmath,amssymb,bm,dsfont,amsthm}
\usepackage{graphicx}

\usepackage[centertableaux]{ytableau}   \usepackage[centertableaux]{ytableau}   
\usepackage{xcolor}
\usepackage{booktabs}             
\allowdisplaybreaks[1]

\usepackage[breaklinks, colorlinks=true, linkcolor=blue, urlcolor=blue,citecolor=blue]{hyperref}

\DeclareFontFamily{U}{mathx}{\hyphenchar\font45}
\DeclareFontShape{U}{mathx}{m}{n}{
      <5> <6> <7> <8> <9> <10>
      <10.95> <12> <14.4> <17.28> <20.74> <24.88>
      mathx10
      }{}
\DeclareSymbolFont{mathx}{U}{mathx}{m}{n}
\DeclareMathSymbol{\bigtimes}{1}{mathx}{"91}

\newtheorem{thm}{Theorem}
\newtheorem{lemma}[thm]{Lemma}
\newtheorem{example}[thm]{Example}
\newtheorem{definition}[thm]{Definition}

\newtheorem{prop}[thm]{Proposition}

\newtheorem{remark}[thm]{Remark}
\newtheorem{proposition}[thm]{Proposition}
\newtheorem*{prop*}{Proposition}

\DeclareMathOperator{\tr}{tr}
\DeclareMathOperator{\SWAP}{SWAP}

\newcommand{\ot}[0]{\otimes}

\newcommand{\nn}[0]{\nonumber}

\newcommand{\one}[0]{\mathds{1}}
\newcommand{\id}[0]{\operatorname{id}}

\newcommand{\bra}[1]{\mathinner{\langle #1|}}
\newcommand{\ket}[1]{\mathinner{|#1\rangle}}

\newcommand{\dyad}[1]{| #1\rangle \langle #1|}

\renewcommand{\a}{\alpha}

\renewcommand{\>}{\rangle}

\newcommand{\overbar}[1]{\mkern 1.5mu\overline{\mkern-1.5mu#1\mkern-1.5mu}\mkern 1.5mu}

\begin{document}

\title{Positive maps from the walled Brauer algebra}

\author{Maria Balanzó-Juandó}
\email{maria.balanzo@icfo.eu}
\thanks{MBJ was supported by 
the European Union 
(Marie Skłodowska-Curie 847517), 
the Government of Spain (FIS2020-TRANQI, Severo Ochoa CEX2019-000910-S), 
Fundaci\'o Cellex, 
Fundaci\'o Mir-Puig, 
and the Generalitat de Catalunya (CERCA, AGAUR SGR 1381).}
\affiliation{ICFO – Institut de Ciencies Fotoniques, The Barcelona Institute of Science and Technology, 08860 Castelldefels (Barcelona), Spain}

\author{Micha\l{} Studzi\'nski}
\thanks{MS acknowledges  support by the National Science Centre grant Sonata 16 (2020/39/D/ST2/01234).}
\affiliation{Institute of Theoretical Physics and Astrophysics and National Quantum Information Centre in Gda{\'n}sk,
		Faculty of Mathematics, Physics and Informatics, University of Gda{\'n}sk, PL-80-952 Gda{\'n}sk, Poland}
		
\author{Felix Huber}
\thanks{FH was supported by the FNP through TEAM-NET (POIR.04.04.00-00-17C1/18-00).}
\affiliation{Atomic Optics Department, Jagiellonian University, PL-30-348 Kraków, Poland}
\affiliation{Bordeaux Computer Science Laboratory (LaBRI), University of Bordeaux, 351 cours de la Liberation, 33405, Talence, France}

\begin{abstract}
    We present positive maps and matrix inequalities 
    for variables from
    the positive cone.
    These inequalities contain partial transpose and reshuffling operations,
    and can be understood as positive multilinear maps
    that are in
    one-to-one
    correspondence with elements from
    the the walled Brauer algebra.
    Exploring the entanglement structure of multipartite states, we relate them to different separability classes
\end{abstract}
\maketitle

\section{Introduction}
Entanglement is one of the defining features of quantum many-body systems,
and its detection and characterisation presents us an ongoing challenge~\cite{
Ketterer2020entanglement, 
PhysRevLett.124.200502,
Neven2021,
PhysRevA.104.042420,
Knips2020,
frerot2021unveiling,
Hiesmayr2021,
Marconi2021entangledsymmetric}.
The key methods for entanglement detection, 
entanglement witnesses and positive maps, 
rely on our understanding of the 
mathematical features of multilinear algebra.
It is thus desirable to establish a dictionary for translating results between the two fields.

In algebraic geometry, various Positivstellensätze exist for polynomials and trace polynomials
whose variables are matrices~\cite{
10.2307/3597203, 
10.2307/3844994,
Helton2012,
klep2020optimization}. However,
few methods are applicable for variables from the positive cone, i.e. for positive semidefinite matrices.

Recently a one-to-one correspondence between
inequalities for matrix trace polynomials and
Werner state witnesses 
was established~\cite{Huber2021}.
One the one hand, the formalism opened the way to perform dimension-free entanglement detection, 
a method that is independent of the local Hilbert space dimension~\cite{https://doi.org/10.48550/arxiv.2108.08720}.
On the other hand, it offers a
systematic computational method to prove matrix inequalities from the positive cone, thus complementing methods from algebraic geometry.

In quantum information and computation, more general matrix operations like the partial transpose, partial trace, and
reshuffling are often considered. 
It is then a natural question to ask: which "polynomials" are positive with respect to
this new set of operations? 
Our manuscript addresses this question for the case such polynomials that are multilinear
and provides key tools for their manipulation. Additionally, it integrates their construction into earlier concepts from
entanglement theory and multilinear algebra.

This enlarges the correspondence from Werner state witnesses to witnesses from the {\em walled Brauer algebra},
formed by partially transposed permutation operators
acting on $(\mathbb{C}^d)^{\ot n}$~\cite{Eggeling2001}.
This algebra is the commutant of the adjoint action of $U^{\ot (n-k)} \ot {\bar{U}}^{\ot k}$, 
with $U$ unitary and $\bar{U}$ its complex conjugate.
Our framework allows to port previously known results on the entanglement of states with partial positive transpose~\cite{Eggeling2001} into the domain of matrix inequalities,
and translates positive multilinear maps~\cite{Bardet2020} back into entanglement witnesses. 
Furthermore, it provides a convenient framework for working with a class of 
nonlinear transformation of quantum states, 
rendering it of interest for nonlinear quantum computation~\cite{holmes2021nonlinear} and equivariant quantum neural networks.

Section~\ref{sec:basicconcepts} introduces the key concepts of this paper. Section~\ref{sec:tripartitesystems} treats examples in tripartite scenarios by
discussing a map by Bardet-Collins-Sapra and the results by Eggeling and Werner on $U^{\ot 3}$-invariant states with positive partial transpose (PPT).
Section~\ref{sec:MultiMapsBlocPos} then explores multipartite maps and their positivity in general,
while Section~\ref{sec:irrepssection} constructs 
positive multilinear maps induced by irreducible representations of the 
commutant of $U^{\ot (n-k)} \ot {\bar{U}}^{\ot k}$. 
Appendix~\ref{appA:rep_th} contains a short summary on the representation theory of the walled Brauer algebra,
Appendix~\ref{app:wernermaps} provides a table of  positive maps, and 
Appendix~\ref{app:proofs} contains some technical proofs. Finally,  Appendix~\ref{app:graphicaltrick} explains graphical notations for working with permutations and their partial transposes.

\section{Basic concepts}\label{sec:basicconcepts}

The set of complex $d\times d$ matrices is  $\mathcal{M}_d$. 
The \emph{positive cone} is the set of complex positive-semidefinite $d\times d$ matrices.
An operator $A$ is positive semidefinite (written as $A\geq 0$) if and only if $\tr(AB) \geq 0$ for all $B \geq 0$ holds. 
This property is also known as the \emph{self-duality of the positive cone}. 
A linear map $\Lambda: \mathcal{M}_{d_1}\to \mathcal{M}_{d_2}$ is \emph{positive} if $\Lambda(X) \geq 0$ for all $X \geq 0$ and $X\in\mathcal{M}_{d_1}$.

A map $\Lambda: \mathcal{M}_{d_1}^r\to \mathcal{M}_{d_2}$ is \textit{multilinear} if it is linear in each variable. We call that a multilinear map $\Lambda$ is positive, if $\Lambda(X_1,\ldots,X_n)\geq 0$ for all $X_1,\ldots,X_n\geq 0$.

\subsection{Matrix operations}\label{sec:tricksandconcepts}

In quantum mechanics, the \emph{partial trace} yields the reduced or local description of quantum states. 
The coordinate-free definition of the partial trace $\tr_1$ states that it is the unique linear operator satisfying
\begin{equation}\label{eq:ptrace_def}
    \tr\big[M(\one\ot N)\big]=\tr\big[\tr_1(M)N\big],
\end{equation}
for all operators $M$ and $N$ acting on Hilbert spaces $\mathcal{H}_1\otimes\mathcal{H}_2$ and $\mathcal{H}_2$, respectively.

We will use permutation operators frequently. These generalise the well-known swap operator $\SWAP$ (in shorthand, $(12)$),
which interchanges the two tensor factors of $\mathbb{C}^d \ot \mathbb{C}^d$ as $(12)\ket{\psi} \ot \ket{\phi} = \ket{\phi}\ot \ket{\psi}$, 
for all $|\phi\rangle,|\psi\rangle\in\mathbb{C}^d$. 
An arbitrary permutation $\sigma \in S_n$ acts on $(\mathbb{C}^d)^{\otimes n}$ as
\begin{equation}
 \sigma \ket{v_1} \ot \dots \ot \ket{v_n} =  \ket{v_{\sigma^{-1}(1)}} \ot \dots \ot \ket{v_{\sigma^{-1}(n)}}\,.
\end{equation}
Throughout the text we will use the cyclic notation for permutations, 
so that $(134)$ maps $1\to 3$, $3\to 4$, and $4\to 1$ while $2 \to 2$ remains. 
For the identity permutation we write $\id$ or $\one$. 
We slightly abuse notation, and
both elements of the symmetric group and their 
representations on $(\mathbb{C}^d)^{\otimes n}$ 
are denoted by $\sigma$.
The distinction will be clear from the context. 

Let us now combine these concepts. For example, the `swap trick' states that 
$ \tr[(12)A\ot B]=\tr(AB)\label{traceSWAP}$ for all square matrices $A,B$ of equal size. 
This idea was generalised in Ref.~\cite{Huber2021} to
\begin{equation}
    \tr_{12}\big[(321)A\ot B\ot C\big]=ABC \label{productmatrices}\,,
\end{equation} 
where now the partial trace translates a permutation into a corresponding matrix product.

The \emph{partial transpose} ${(\cdot)}^{T_1}$
is defined as the linear extension of the ordinary matrix transposition $(\cdot)^T$ with respect to a given basis $\langle i|A^T |j\rangle := \bra{j}A \ket{i}$. Namely, we have 
\begin{equation}
    T_1:\;|\underline{i} \rangle\langle \underline{j}|\otimes |k\rangle\langle l|\mapsto |\underline{j}\rangle\langle \underline{i}|\otimes |k\rangle\langle l|\,.
\end{equation}
With respect to the second subsystem the partial transpose is defined analogously. 
The partial transpose relates the swap operator to the Bell state $\ket{\phi^+}=\frac{1}{\sqrt{d}}\sum_{i=1}^d\ket{ii}$ through $(12)= d  \dyad{\phi^+}^{T_1}$.

New interesting relations appear when considering partial transposes. 
For instance,
\begin{equation}\label{eq:concept_ptranspose_swap}
    \tr_{1}\big[(12)^{T_1}A\ot B \big]=A^TB\,.
\end{equation}
That is, a partially transposed permutation operator translates into 
the transposition of a variable in the output product. 

When dealing with more tensor factors, one obtains the {\em reshuffling}, also known as realignment, which can be seen as transposition 
that involves the indices of different tensor factors.
The reshuffling $(\cdot)^R$ interchanges the bra of the first tensor factor with the ket of the second tensor factor, 
\begin{equation}
    R:\;|i \rangle\langle \underline{j}|\otimes |\underline{k}\rangle\langle l|\mapsto |i\rangle\langle \underline{k}|\otimes |\underline{j}\rangle\langle l|\,.
\end{equation}
An example how the this operation appears in a linear maps is
\begin{equation}
    \tr_{1}\big[(123)^{T_2}A\ot \one\ot\one\big]=(\one\ot A)^R\,.
\end{equation}
This reshuffling operation appears in the reshuffling criterion to detect entanglement
(also known as realignment or computable cross-norm criterion) \cite{Chen2003, Rudolph2005}.

Going beyond $3$ tensor factors, an example is
\begin{equation}\label{eq:re3}
    (A^R B^R)^R = \tr_{23}\big[ \one \ot (23)^{T_3} \ot \one (A \ot B) \big].
\end{equation}

A graphical representation of the partial trace, partial transposition, and reshuffling operations can be seen in Figure~\ref{fig:partialtransposition.pdf}.

\subsection{Entanglement}

A quantum state (density matrix) is a positive semidefinite hermitian matrix $\varrho$ of trace one. 
A bipartite state $\varrho^{AB}$ is called \emph{separable} (or classically correlated) 
if it can be written as a convex combination of product states, 
\begin{equation}\label{eq:sep}
    \varrho^{AB}=\sum_j p_j \varrho_j^A\otimes \varrho_j ^B\,, 
\end{equation}
where $p_i \geq 0\,, \sum_j p_j = 1$  and $\varrho_j^A$, $\varrho_j^B$ are density matrices on systems $A$ and $B$ respectively. 
When a state cannot be written as in Eq.~\eqref{eq:sep} is called \emph{entangled}. 

Multipartite systems show a richer entanglement structure. 
In the tripartite scenario,
a state of three particles $\varrho^{ABC}$ is \emph{fully separable} if it can be written as
\begin{equation}\label{fullyseparable}
    \varrho^{ABC}=\sum_j p_j \varrho^{A}\otimes \varrho^{B} \otimes \varrho^{C}.
\end{equation}
If $\varrho^{ABC}$ cannot be written in this form, it is entangled. While some states cannot be written as in Eq.~\eqref{fullyseparable}, they admit a biseparable decomposition across the bipartition $A|BC$
\begin{equation}
    \varrho^{A|BC}=\sum_j p_j \varrho^{A}\otimes \varrho^{BC}\,.
\end{equation}
Separability across the bipartitions $B|AC$, and $C|AB$
is defined analogously.  A state is said to be \emph{biseparable} if it can be written as a convex combination of biseparable states,
\begin{equation}
    \varrho^{ABC}=p_a \varrho^{A|BC}+p_b \varrho^{B|AC}+p_c \varrho^{C|AB}\,.
\end{equation}
A state $\varrho^{ABC}$ is  \emph{genuinely multipartite entangled} if it is not biseparable. 
It is interesting to note that there exist states which are biseparable for every bipartition, but that are not fully separable \cite{Bennett99,Acin2001}.
These different separability classes are shown in Fig.~\ref{fig:tri_sets.pdf}.

\begin{figure}[t]
  \centering
  \includegraphics[width=\columnwidth]{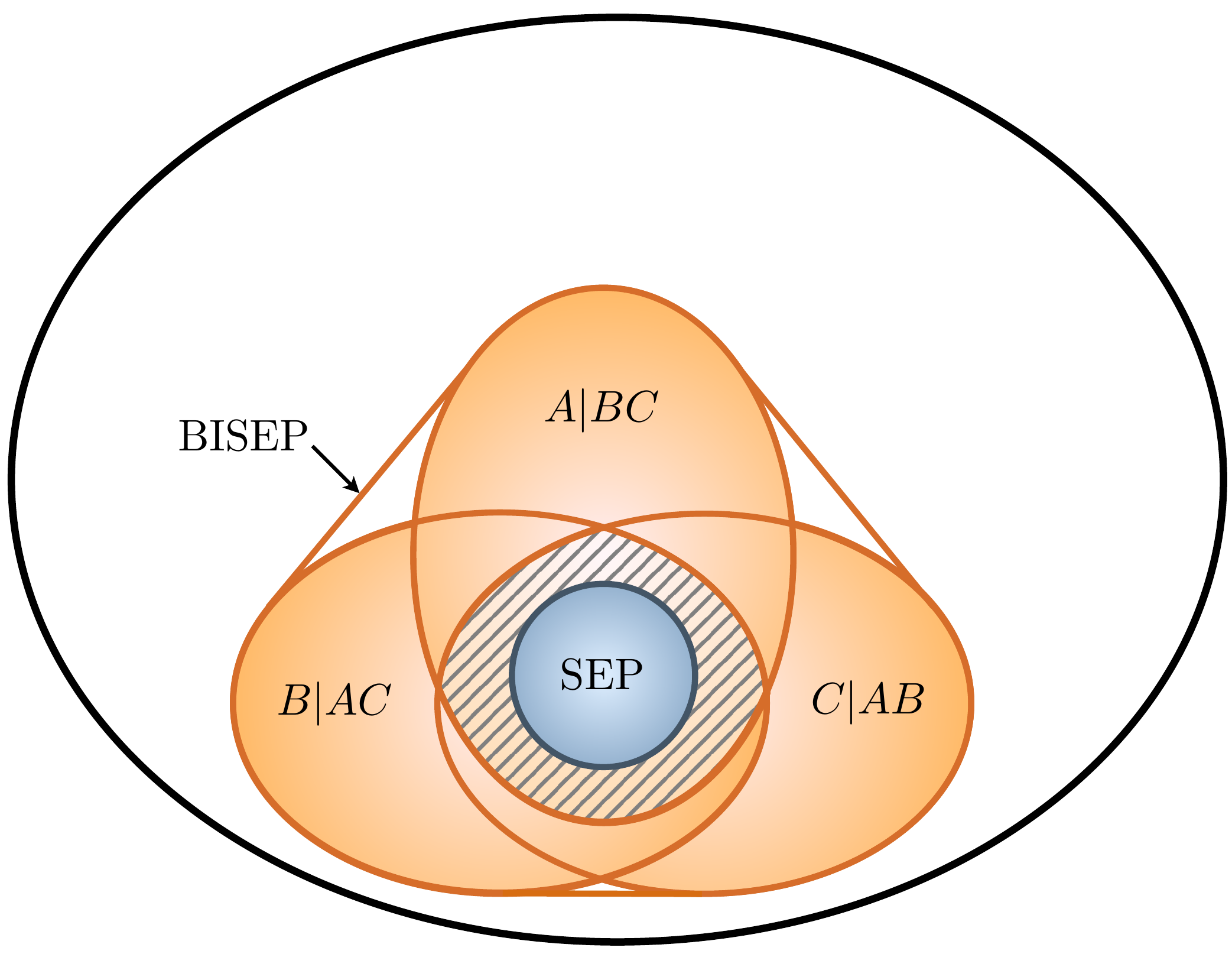}
  \caption{
  Separability classes of all tripartite states. The set of fully separable states $\text{SEP}$ (blue) is strictly included in the union of the biseparable sets $A|BC\,,B|AC\,,C|AB$. Thus there are states that are biseparable for every possible bipartition, 
  but which are not fully separable \cite{Bennett99,Acin2001} (shaded region). States belonging to the convex hull of $A|BC\,,B|AC\,,C|AB$ are called biseparable, $\text{BISEP}$.}
  \label{fig:tri_sets.pdf}
\end{figure}

Determining whether a state is entangled or not is a computationally difficult task~\cite{GURVITS2004448}. 
In the bipartite scenario, a necessary condition for a bipartite density matrix state to be separable
is the positivity of its partial transpose~\cite{Peres1996}; 
a criterion that is also sufficient in $\mathbb{C}^2 \ot \mathbb{C}^2$ and $\mathbb{C}^2 \otimes \mathbb{C}^3$~\cite{Horodeckis1996}.

A more general method for entanglement detection is to use entanglement witnesses~\cite{Terhal2000, Chru_ci_ski_2014}. These are operators $\mathcal{W}$ for which $\tr(\mathcal{W}\varrho)\geq 0$ holds for all separable states $\varrho$, and $\tr(\mathcal{W}\sigma)< 0$ for some entangled state $\sigma$. In words, an observable $\mathcal{W}$ is an \emph{entanglement witness} if its expectation value is non-negative on the set of separable states, while having at least one negative eigenvalue. However, no systematic analytical method is known for constructing entanglement witnesses, as this would amount to being able to solve the separability problem.

An operator $\mathcal{V}$ is a \emph{block-positive} if $\tr(\mathcal{V}\varrho)\geq 0$ for all separable states. Similarly, it is block-positive for a partition $A_1|A_2|\ldots|A_\ell$ if
\begin{equation}
 	\bra{\phi_{A_1}} \bra{\psi_{A_2}} \dots \bra{\phi_{A_\ell}} \mathcal{V}\ket{\phi_{A_1}} \ket{\psi_{A_2}} \dots \ket{\phi_{A_\ell}} \geq 0\,,
\end{equation}
for all $\ket{\phi_{A_1}},\ket{\psi_{A_2}}, \dots, \ket{\phi_{A_\ell}}$.

The Choi-Jamio{\l}kowski isomorphism relates linear maps to operators,
\begin{equation}
    \Lambda_\Phi(\varrho) = \tr_1\big[ \Phi (\rho^T \ot \one) \big]\,.
\end{equation}
It is known that $\Lambda_\Phi$ is completely positive if and only if $\Phi$ is positive semidefinite.
It is easy to see that a similar relation holds between positive maps and bipartite block-positive operators through
\begin{align}\label{eq:JI_wit}
\tr\big[ \Lambda_\Phi (\varrho) C\big] 
&= 
\tr \Big[
\tr_1\big( \Phi (\rho^T \ot \one) \big) C 
\Big]  \nonumber\\
&=\tr \big[ \Phi (\rho^T \ot C) \big] \,,
\end{align}
where we used the defining property of the partial trace~\eqref{eq:ptrace_def}.
By Eq.~\eqref{eq:JI_wit} and the self-duality of the positive cone, $\Lambda$ is a positive map if and only if $\Phi$ is block-positive (that is, $\Phi$ either a witness or positive semidefinite). 

This relation extends to multilinear positive maps and multipartite block-positive operators.
Write $\Lambda(X_1,\dots,X_n)$ as,
\begin{align}\label{eq:multilinearmap}
    \Lambda(X_1,\dots,X_n)
    =
    \tr_{1\dots n}
    \big[\tilde\Phi(X_1\ot\dots\ot X_n\nonumber\\\ot\one\ot\dots\ot\one)
    \big]\,,
\end{align}
for some operator $\tilde\Phi$.
From
\begin{align}\label{eq:JI-multipartite}
    &\phantom{=} \tr\Big[
    \Lambda(X_1,\dots,X_n)C \Big] \nonumber \\
    &=
    \tr\Big[
    \tr_{1\dots n}
    \big[
    \tilde\Phi(X_1\ot\dots\ot X_n
    \ot\one\ot\dots\ot\one) 
    \big]
    \Big]\nonumber\\
    &=\tr\big[
    \tilde\Phi(X_1\ot\dots\ot X_n\ot C)
    \big]\,,
\end{align}
one sees that 
$\Lambda(X_1,\dots,X_n)$ is a multilinear positive map,  
if and only if 
the operator $\tilde\Phi$ is block-positive. 
We will use this trick throughout the paper.

\subsection{Walled Brauer algebra}
The representation theories of the symmetric group $S_n$ and the general linear group $GL_d(\mathbb{C})$ are related by the famous Schur-Weyl duality~\cite{Goodman}. Taking $V=\mathbb{C}^d$, the Schur-Weyl duality states that there exists a basis in which one has the following decomposition
\begin{align}
\label{wedd1}
(\mathbb{C}^d)^{\otimes n}=\bigoplus_{\alpha} \mathcal{U}_{\alpha}\otimes \mathcal{S}_{\alpha}\,,
\end{align}
where the symmetric group $S_n$ acts on the representation space $\mathcal{S}_{\alpha}$, and the general linear group $GL_d(\mathbb{C})$ acts on the representation space $\mathcal{U}_{\alpha}$, labelled by the same irreducible representation $\alpha$. 
From this decomposition it follows that the diagonal action of the general linear group $GL_d(\mathbb{C})$ of invertible complex matrices and of the symmetric group $S(n)$ on $V^{\otimes n}$ commute. 

We can identify the elements from $S_n$ as a permutation operators as it is described at the beginning of section~\ref{sec:tricksandconcepts}. 
It turns out that permutation operators can be represented and manipulated graphically, see Figure~\ref{fig:WBA}. 

However, there exists also another version of the Schur-Weyl duality exploiting the commutation relation of the natural representation of $GL_d(\mathbb{C})$ on mixed tensor product $V^{\ot n}\otimes \bar{V}^{\otimes (n-k)}$ of and elements forming so called walled Brauer algebra. This algebra has been in introduced by Tuarev~\cite{Tur}, Koike~\cite{Ko}, and Benkart et al~\cite{BEN}. 

The walled Brauer algebra $\mathcal{WBA}(n,d,k)$ with $n>k$ acting on $V=\mathbb{C}^d$ is formed by linear combinations 
of permutation operators 
that are partially transposed over last $k$ systems,
\begin{equation}
    \mathcal{WBA}(n,d,k) = \left\{ \sum_{\sigma \in S_n} a_\sigma \sigma^{T_{k\dots n}} \,|\, a_\sigma \in \mathbb{C} \right\}.
\end{equation}
Similarly to the permutation operators, the $\mathcal{WBA}$ 
operators admit a graphical representation which makes them easy to manipulate (see Figure~\ref{fig:WBA}).

\begin{figure}[tbp]
\begin{center}
   \includegraphics[width=0.8\columnwidth]{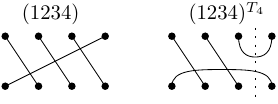}
\caption{Graphical representation of the permutation operator $(1234)$ --left-- and the representation of the partially transposed permutation operator $(1234)^{T_4}$ --right--. The vertical dotted line depicts the 'wall' in the walled Brauer algebra. This wall separates transposed elements from the untouched ones.  For more details on this graphical representations see Appendix \ref{app:graphicaltrick}.}
\label{fig:WBA}
\end{center}
\end{figure}

While superficially similar, the partial transposition makes the group algebra 
$\mathbb{C}[S_n] 
= \big\{ \sum_{\sigma \in S_n} a_\sigma \sigma \,|\, a_\sigma \in \mathbb{C} \big\}$ and $\mathcal{WBA}(n,d,k)$ structurally quite different. 
This can be seen by considering the SWAP operator $(12) \in \mathbb{C}[S_2]$. 
The SWAP is its own inverse element and forms together with the neutral element the symmetric group $S_2$. 
Partially transposing $(12)$ yields an operator $(12)^{T_2}$ proportional to the projector on maximally entangled state
and belongs to $\mathcal{WBA}(2,d,1)$.
One can see that $(12)^{T_2}(12)^{T_2}=d(12)^{T_2}$, 
implying that that $\mathcal{WBA}(2,d,1)$ is no longer a group algebra. A similar argument applies to an arbitrary $\mathcal{WBA}(n,d,k)$. 
Further details on the walled Brauer algebra can be found in Appendix~\ref{appA:rep_th} and references therein.

\section{Tripartite scenarios}\label{sec:tripartitesystems}
Here we illustrate our formalism in the tripartite setting,
relating it to $U \ot U \ot U$-invariant PPT states~\cite{Eggeling2001} and a recently introduced class of positive maps~\cite{Bardet2020}. 
In particular, we illustrate that the considered maps correspond to elements of $\mathbb{C}[S_n]$ and $\mathcal{WBA}(n,d,k)$,
motivating us to study the 
maps generated by these algebras further.

\subsection{The Bardet-Collins-Sapra map}
How does one turn a positive map into a witness for {\em tripartite} states? Consider the map introduced by Bardet, Collins, and Sapra~\cite{Bardet2020},
\begin{align}
    \Phi_{\alpha,\beta,d}: \mathcal{M}_d(\mathbb{C})&\to \mathcal{M}_d(\mathbb{C})\otimes \mathcal{M}_d(\mathbb{C})\nonumber\\
    A &\mapsto A^T \otimes \one_d + \one_d\otimes A\nonumber\\
    &\quad\,\, +\tr(A)\big(\alpha \one_d \ot \one_d +\beta \phi^+\big)\,,
\end{align}
where $\alpha,\beta\in\mathbb{R}$, $d\geq 3$, and $\phi^+$ is the unnormalised Bell state. The authors show that $\Phi_{\alpha,\beta,d}$ is a positive map if and only if the following conditions hold
\begin{enumerate}
    \item \label{itm:firstcondition}$\alpha\in[0,\infty)$ if $\beta\geq0$,
    \item \label{itm:secondcondition}$\alpha\in\left[\frac{-(2+d\beta)+\sqrt{d^2\beta^2-4(d-2)\beta + 4}}{2},\infty\right)$ if $\beta\leq 0$.
\end{enumerate}

Using Eq.~\eqref{productmatrices} and Eq.~\eqref{eq:concept_ptranspose_swap}
from Section~\ref{sec:tricksandconcepts}, we can write this map as
\begin{align}
    \Phi_{\alpha,\beta,d}:\quad \mathcal{M}_d(\mathbb{C})&\to \mathcal{M}_d(\mathbb{C})\otimes \mathcal{M}_d(\mathbb{C})\nonumber\\
    A &\mapsto \tr_1 (P(\alpha,\beta) A\otimes \one_d\otimes \one_d)\,,
\end{align}
where $P(\alpha,\beta)$ is a linear combination of permutation operations and their partial transposes,
\begin{align}
    P(\alpha,\beta)=(12)^{T_2}+(13)+\alpha \id+\beta (23)^{T_2}\,.
    \label{P(a,b)}
\end{align}
When $\alpha$ and $\beta$ take values satisfying conditions ~\ref{itm:firstcondition} and~\ref{itm:secondcondition}, the operator $P(\alpha,\beta)$ is block-positive in the bipartition $1|23$. 
To see that we use an argument as in Eq.~\eqref{eq:JI_wit}: the self-duality of the positive cone states that an operator $A$ is positive if and only if $\tr(AC)\geq 0$ holds for all $C\geq 0$. Therefore
\begin{align}
\label{eq:BCS_proof}
    \tr\left[\tr_1 (P(\alpha,\beta) A\otimes \one_d\otimes \one_d) C\right] \nn\\
    =\tr\left[(P(\alpha,\beta) A\otimes C \right]\geq 0\,,
\end{align}
where we have used the coordinate-free definition of the partial trace. Note that $C$ in Eq.~\eqref{eq:BCS_proof} acts on $\mathbb{C}^d \ot \mathbb{C}^d$.

To obtain an entanglement witness from  $\Phi_{\alpha,\beta,d}$, we need to determine the values $\alpha, \beta$ for which the operator $P(\alpha,\beta)$ is {\em not} positive semidefinite, but fulfills the conditions for the map to be positive. This yields an entanglement witness for the bipartition $1|23$.

Let us construct a witness: consider condition \ref{itm:secondcondition} for the case $d=3$. The map $\Phi_{\alpha,\beta,d}$ is positive for
\begin{equation}
    \alpha\in\left[\frac{-(2+3\beta)+\sqrt{9\beta^2-4\beta+4}}{2},\infty\right),\quad \beta\leq 0\,.
\end{equation}
Fig.~\ref{fig:Negativity} shows the values of $\alpha$ and $\beta$ for which the operator $P(\alpha,\beta)$ acts as a witness.

\begin{figure}[tbp]
\begin{center}
   \includegraphics[width=\columnwidth]{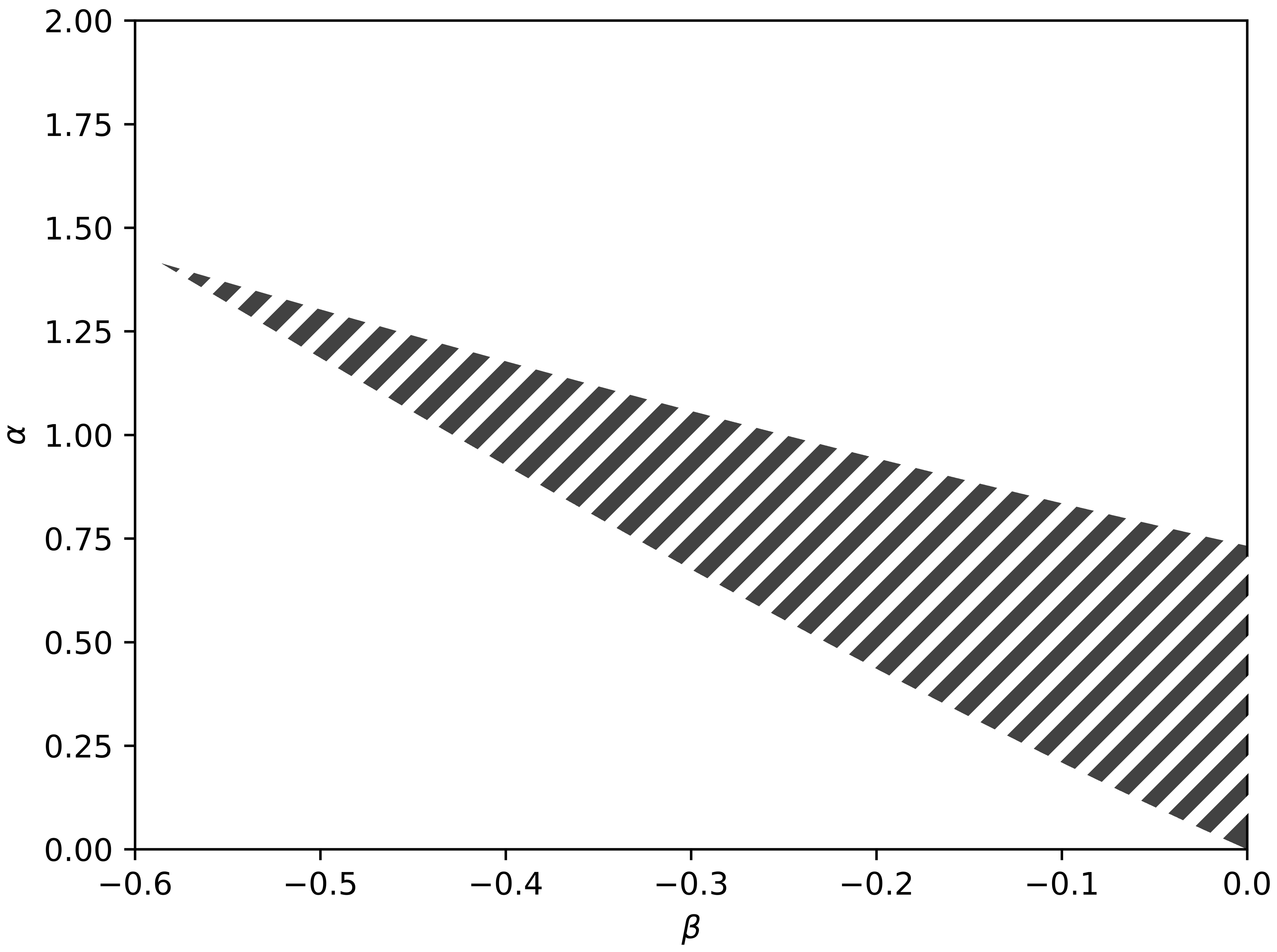}
\caption{The striped region represents the pairs of $(\alpha,\beta)$ for which $P(\alpha,\beta)$ of Eq.~\eqref{P(a,b)} is a witness of entanglement in the bipartition $1|23$. That is, these are the points for which the map is positive according to condition \ref{itm:secondcondition} of \cite{Bardet2020}, and for which the operator $P(\alpha,\beta)$ is not positive semidefinite.}
\label{fig:Negativity}
\end{center}
\end{figure}
For instance, taking $\beta=-1/10$
condition~\ref{itm:secondcondition} requires that
$\alpha\in[\approx 0.21,\infty)$. 
A possible value is $\alpha=1/4$, and the operator $P(1/4,-1/10)$ then reads 
\begin{equation}
    P(1/4,-1/10)=(12)^{T_2}+(13)+\frac{1}{4}\id-\frac{1}{10}(23)^{T_2}\,.
    \label{Pwitness}
\end{equation}
Clearly, this operator is non-zero, and by the preceding discussion acts an entanglement witness on the bipartition $1|23$ for the $d=3$ case.

\subsection{Eggeling-Werner maps} 
Here we give another tripartite example and show how to derive matrix inequalities from states with a positive partial transpose (PPT states). 

Every $U\ot U\ot U$-invariant state $\varrho$ (also known as Werner states) on $(\mathbb{C}^d)^{\ot 3}$ can be written as a linear combination of permutation operators~\cite{Eggeling2001}. Namely,
\begin{align}\label{eq:werner_para}
    \varrho&=\alpha_1\id+\alpha_2(12)+\alpha_3(23) +\alpha_4(31)\nn\\
    &\quad\,\, +\alpha_5(123)+\alpha_6 (321)\nn\\
    &=\sum_{k=+,-,0,1,2,3} c_k R_k\,,
\end{align}
where $R_k$ are the following operators
\begin{align}\label{Rks}
    R_+&=\frac{1}{6}\left[\id+(12)+(23)+(31)+(123)+(321)\right],\nn\\
    R_-&=\frac{1}{6}\left[\id-(12)-(23)-(31)+(123)+(321)\right],\nn\\
    R_0&=\frac{1}{3}\left[2\cdot\id-(123)-(321)\right],\nn\\
    R_1&=\frac{1}{3}\left[2\cdot(23)-(31)-(12)\right],\nn\\
    R_2&=\frac{1}{\sqrt{3}}\left[(12)-(31)\right],\nn\\
    R_3&=\frac{i}{\sqrt{3}}\left[(123)-(321)\right]\,.
\end{align}
Admissible values of coefficients $\alpha_k$ and $c_k$ for $\varrho$ to form a density matrix can be found in the Appendix~\ref{app:wernermaps}. 

The below Lemma from  Ref.~\cite{Eggeling2001} states the conditions for a Werner state to have a positive partial transpose.
\begin{lemma}[Lemma 10 of Ref.~\cite{Eggeling2001}]\label{wernerlemma}
Let $\varrho$ be a $U^{\ot 3}$-invariant state with $r_k=\tr(\varrho R_k)$, $k=+,-,1,2,3.$ Then 
$\varrho^{T_1}\geq 0$ if and only if
\begin{subequations}
\begin{align}
    0&\leq r_-\\
    0&\leq r_1-r_+-r_-+1   \\
    0&\leq 1-r_1-5r_--r_+  \allowdisplaybreaks\\
    0&\leq -1-r_1+r_-+5r_+    \allowdisplaybreaks\\
    r_2^2+r_3^2&\leq F_1 \\
    r_2^2+r_3^2&\leq F_2\,,
\end{align}
\label{conditionsrhopositive}
\end{subequations}
where $F_1=(1-r_1-5r_--r_+)(-1-r_1+r_-+5r_+)/3$ and $F_2=(1-r_1-r_--r_+)(1+r_1-r_--r_+)/3$.
\end{lemma}
\begin{remark}
The relation between $r_k$ and the coefficients $c_k$ and $\alpha_k$ of Eq.~\eqref{eq:werner_para} is given in Appendix~\ref{app:wernermaps}.
\end{remark}

To a $U^{\ot 3}$-invariant $\varrho$ we associate the following maps,
parameterised by $\a_1, \dots, \a_6$ of Eq.~\eqref{eq:werner_para},
\begin{align}\label{Wernermaps}
    f_S(A)   &=\tr_1[\varrho^{T_S}A\ot\one\ot\one] \,, \\
    g_S(A,B) &=\tr_{12}[\varrho^{T_S}A\ot B\ot\one]\,,\label{Wernermaps2}
\end{align}
where $T_S$ denotes partial transposition over $S \subsetneq \{1,2,3\}$.
Below we give examples how these maps look in terms of the discussed matrix operations.

\begin{example}
The partial traces and partial transposes involved can act on one or two subsystems,
\begin{align}
    f_3(A) &= \tr_1[\varrho^{T_3} A\ot\one\ot \one]\nn\\
    &=\alpha_1\tr(A)(\one\ot\one)+\alpha_2(A\ot\one)\nn\\
    &+\alpha_3\tr(A)(\one\ot\one)^R +\alpha_4(\one\ot A^T)\nn\\
    &+\alpha_5(A\ot\one)^R+\alpha_6(\one\ot A^T)^R\,,  \nn\\
    \nn\\
   f_{23}(A)  &= \tr_{1}[\varrho^{T_{23}} A \ot \one \ot \one ] \nn\\ 
   &=\alpha_1\tr(A)(\one\ot\one)+\alpha_2(A^T\ot\one)\nn\\
  &+\alpha_3\tr(A)((\one\ot\one)^R)^{T_2}+\alpha_4(\one\ot A^T)\nn\\
  &+\alpha_5((\one\ot A)^R)^{T_2}+\alpha_6((A^T\ot\one)^R)^{T_2}\,, 
  \nn\\
  \nn\\
   g_{3}(A,B)  &=  \tr_{12}[\varrho^{T_3}  A \ot B \otimes \one] \nn\\
   &= \alpha_1\tr(A)\tr(B)\one+\alpha_2\tr(AB)\one\nn\\
   &+\alpha_3\tr(A)B^T+\alpha_4\tr(B) A^T\nn\\
   &+\alpha_5A^TB^T+\alpha_6B^TA^T\,.
\end{align}
\end{example}
To obtain these maps, we have used the graphical notation explained in Appendix~\ref{app:graphicaltrick}.
Depending on the overlap between partial transpose and partial trace, different operations appear in the resulting map.
A table of all maps of this type is in the Appendix~\ref{app:wernermaps}. 

Can these maps be made stronger? 
In analogy to Eq.~\eqref{eq:JI_wit}, one sees that a block-positive operator suffices.
\begin{proposition}\label{wernermapspositive}
The maps $f_S$ are positive if and only if $\varrho^{T_S}$ is block-positive with respect to the partition $1|23$.
The maps $g_S$ are positive if and only if $\varrho^{T_S}$ is block-positive.
\end{proposition}

\begin{proof}
One requires that 
\begin{align}
    \tr[f_S C]&=\tr\big[\tr_{1}(\varrho^{T_S} A\ot \one\ot \one)C\big]\nn\\
    &=\tr[\varrho^{T_S} A\ot C_{23}]\geq 0, \qquad \forall C_{23}\geq 0\,.
\end{align}
This is the case if and only if $\varrho^{T_S}$ is block-positive in the bipartition $1|23$. Similarly, $g_S$ is positive if and only if
\begin{align}
    \tr[g_S C]&=\tr\big[\tr_{12}(\varrho^{T_S} A\ot B\ot \one)C\big]\nn\\
    &=\tr[\varrho^{T_S} A\ot B\ot C]\geq 0, \qquad \forall C\geq 0 \,,
\end{align}
corresponding to an operator  $\varrho^{T_S}$ that is block-positive for the partition $1|2|3$.
\end{proof}

From here, one can easily see that if $f_S$ is a positive map but $\varrho^{T_1}\not\geq 0$, then $\varrho^{T_1}$ is an entanglement witness across the cut $1|23$. If we now consider the map $g_S$ to be positive, then $\varrho^{T_1}$ is an entanglement witness across $1|2|3$.
A operator that is block-positive for the partition $1|2|3$ is also block-positive for $1|23$. 
Thus if $g_S$ is positive
(for some given $\varrho$), then so is $f_S$. 

\section{Multilinear maps and blockpositivity}
\label{sec:MultiMapsBlocPos}
We now treat a wider setting that involves more than three tensor factors. 
Consider the map $\Lambda(\varrho) = \tr_{12}[P\varrho]$ where $\varrho \in \mathcal{M}_{d^4}$ and $P \in \mathcal{M}_{d^4}$.
If the input $\varrho$ has tensor product structure, it is clear from the previous discussion that $\Lambda$ is positive if and only if $P$ is either a positive operator, or an entanglement witness for the partition $1|2|34$.
In contrast, for entangled input $\varrho$ the map $\Lambda$ is positive if and only if $P$ is block-positive across the bipartition $12|34$. 

Thus it becomes clear that properties required of $P$ depend on the type of input considered.
The following generalises the above discussion.
\begin{proposition}\label{prop:general_pos_map}
Let $X_1,\ldots,X_{k-m}\in \mathcal{M}_d$ and $X_1,\ldots,X_{k-m}\geq 0$. A map 
\begin{align}
&\Lambda(X_1,\dots, X_{k-m}) \nn\\
&= \tr_{1\dots(k-m)}\big[P (X_1 \ot \dots \ot X_{k-m} \ot \one^{\ot m })\big]\,,
\end{align} 
is 
positive if and only if $P$ is block-positive across the partition $1|2|\dots|k-m|k-m+1\dots m$. 
\end{proposition}

\begin{proof}
The proof follows from the self-duality of the positive cone.  
Write $\Lambda(X_1,\dots,X_{k-m})$ as
\begin{align}\label{eq:multilinearmap2}
    \Lambda(X_1,\dots,X_{k-m})
    =
    \tr_{1\dots (k-m)}
    \big[P(X_1\ot\dots\ot X_{k-m}\nonumber\\\ot\one\ot\dots\ot\one)
    \big]\,.
\end{align}
From 
\begin{align}\label{eq:JI-multipartite2}
    &\phantom{=} \tr\Big[
    \Lambda(X_1,\dots,X_{k-m})C \Big] \nonumber \\
    &=
    \tr\Big[
    \tr_{1\dots (k-m)}
    \big[
    P(X_1\ot\dots\ot X_{k-m}
    \ot\one\ot\dots\ot\one) 
    \big]
    \Big]\nonumber\\
    &=\tr\big[
    P(X_1\ot\dots\ot X_{k-m}\ot C)
    \big]\,, \quad \  \forall C\geq 0,
\end{align}
one sees that 
$\Lambda(X_1,\dots,X_{k-m})$ is a multilinear positive map  
if and only if 
$P$ is block-positive.
\end{proof}

How can such maps be evaluated? We state an example before giving general recipes.
\begin{example}[$3\to 2$ tensor factors map]
Consider the following map, where a generic $X_i$ can be written as $X_i=\sum_{\alpha_i,\beta_i=1}^d\chi_{\alpha_i,\beta_i}^{(i)}|\alpha_i\rangle\langle \beta_i|$ 
\begin{align}
    &\tr_{123}[(12345)^{T_2} X_1\ot X_2\ot X_3\ot \one\ot \one]\nn\\
    &=\tr_{123}[(|lmmjk\rangle\langle iijkl|)X_1\ot X_2\ot X_3\ot \one\ot \one]\nn\\
    &=\delta_{\alpha_1\alpha_2}\delta_{\beta_2\beta_3}\chi_{\alpha_1,\beta_1}^{(1)}\chi_{\alpha_2,\beta_2}^{(2)}\chi_{\alpha_3,\beta_3}^{(3)}|\alpha_3\rangle\langle k|\ot|k\rangle\langle \beta_1|\nn\\
    &=((X_3X_2^TX_1\ot\one)^R)^{T_2}\,,
\end{align}
where the last equality follows from the definition of the reshuffling operator and the partial transpose.
\end{example}

These tensor factors can also be thought of as quantum registers. 
In turns out that maps from $k$ to $1$ and $1$ to $k-1$ tensor factors have a particularly nice structure and can be evaluated a quick manner.

\subsection{From $k$ to $1$ registers}
Let us consider maps from $k$ to $1$ tensor factors. 
To see how this works, let $X_i=\sum_{\alpha_i,\beta_i=1}^d\chi_{\alpha_i,\beta_i}^{(i)}|\alpha_i\rangle\langle \beta_i|$ and consider 
\begin{align} \label{eq:4to1}
    \tr_{1234}&\big[(54321)^{T_5} X_1 \ot X_2 \ot X_3\ot X_4\ot X_5\big] \nn\\
    =&\tr_{1234}\big[|jklmm\rangle\langle ijkli| X_1 \ot X_2 \ot X_3\ot X_4\ot X_5\big]\nn\\
    =&\delta_{\alpha_1\alpha_5}\delta_{\beta_1\alpha_2}\delta_{\beta_2\alpha_3}\delta_{\beta_3\alpha_4}\chi_{\alpha_1,\beta_1}^{(1)}\chi_{\alpha_2,\beta_2}^{(2)}\nn\\
    &\cdot\chi_{\alpha_3,\beta_3}^{(3)}\chi_{\alpha_4,\beta_4}^{(4)}\chi_{\alpha_5,\beta_5}^{(5)}|\beta_4\rangle\langle \beta_5|\nn\\
    =&(X_1X_2X_3X_4)^T X_5\,,
\end{align}
where we have used the graphical notation explained in Appendix~\ref{app:graphicaltrick}.
This generalises as follows.
\begin{widetext}
\begin{prop}\label{prop:k-1to1}
Let $X_1,\dots,X_k\in \mathcal{M}_d$ and consider the cycle $(k\dots 1)=(1\dots k)^{-1}$. Then
\begin{align}
    \tr_{1\dots k\backslash k}[(k\dots 1)^{T_j} X_1\ot\cdots\ot X_k]=
        \begin{cases}
            X_1\cdots X_j^T\cdots X_{k-1} X_k &\text{for } j\neq k. \\
            (X_{1}\cdots X_{k-1})^T X_k &\text{for } j=k.
        \end{cases}\\
\allowdisplaybreaks
\tr_{1\dots k\backslash 1}[(1\dots k)^{T_j} X_1\ot\cdots\ot X_k]=
        \begin{cases}
             X_kX_{k-1}\cdots X_j^T \cdots X_1  &\text{for } j\neq 1. \\
              X_k^TX_{k-1}\cdots X_2 X_1  &\text{for } j=1.
        \end{cases}
\end{align}
\end{prop}

The proof can be found in Appendix \ref{app:proofs}.
Consider now permutation operators where multiple indices are partially transposed.
We need the following.
\begin{definition}
Let $\theta_j: \mathcal{M}_d^r \rightarrow \mathcal{M}_d$ act on a tuple of  matrices $(X_{\alpha_1}, \dots, X_{\alpha_r})$ as
\begin{equation}
    \theta_j(X_{\alpha_1}, \dots, X_{\alpha_r}) = 
    \begin{cases}
    X_{\alpha_1} \dots X_j \dots X_{\alpha_r} & \text{if } j \not \in {\alpha_1, \dots, \alpha_r}. \\
    X_{\alpha_1} \dots X_j^T \dots X_{\alpha_r} & \text{if } j  \in {\alpha_1, \dots, \alpha_r}.
    \end{cases}
\end{equation}
Naturally, $\theta_i \theta_j = \theta_j \theta_i$ for all $i,j$.
Similarly, let $\bar\theta_j$ act on the matrix product $X_{\alpha_1} \dots X_{\alpha_r}$ as
\begin{equation}
   \bar \theta_j(X_{\alpha_1}, \dots, X_{\alpha_r}) = 
    \begin{cases}
        X_{\alpha_1} \dots X_j \dots X_{\alpha_r} & \text{if } j  \in {\alpha_1, \dots, \alpha_r}\,. \\
        X_{\alpha_1} \dots X_j^T \dots X_{\alpha_r} & \text{if } j  \not\in {\alpha_1, \dots, \alpha_r}\,.
    \end{cases}
\end{equation}
\end{definition}

With these definitions, we can now consider permutations where multiple indices are partially transposed.
Let $\theta_S = \prod_{j \in S} \theta_j$ and likewise for $\overbar{\theta}_S$.
Recall that $T_S$ denotes the partial transpose on a subset $S$.

\begin{prop}
Let $X_1,\dots,X_k\in \mathcal{M}_d$ and consider the cycle $(k\dots 1)=(1\dots k)^{-1}$. Then
\begin{align}
    \tr_{1\dots k\backslash k}[(k\dots 1)^{T_S} X_1\ot\cdots\ot X_k]=
        \begin{cases}
            \theta_S(X_{1},X_2,\cdots, X_{k-1}, X_k)  &\text{for } k \not\in S\,. \\\bar\theta_S (X_{k-1},X_{k-2},\cdots ,X_{1}, X_k) &\text{for } k \in S\,.
        \end{cases} 
\end{align}
\end{prop}
The proof of this Proposition follows from the previous one applying the partial transpose on all the variables of the subset, instead of only partially transposing one variable.
This can be seen as a generalisation of Proposition~2 in Ref.~\cite{Huber2021} with partial transpose on subset $S$.
\end{widetext}

\subsection{From $1$ to $k-1$ registers}
Here we consider the case from $1$ to $k-1$ registers of the form $\Lambda(A) = \tr_{1}[(1\dots k)^{T_k} A\ot \one\ot\dots\ot \one]$.
This type of map can also be evaluated nicely.
For this we define the reshuffling of the sites $k$ and $l$ in a manner slightly different to the one defined for two sites only 
(c.f. Section~\ref{sec:basicconcepts}).

\begin{definition}[Reshuffling $R_{k,l}$ of sites $k,l$]\label{def_reshuffling} 
Let $X=\ket{i_1 \dots \underline{i_k} \dots i_n} \bra{j_1 \dots \underline{j_l} \dots j_n}\in \mathcal{M}_d^{\otimes n}$. 
Then 
\begin{equation}
    X^{R_{k,l}}=\ket{i_1 \dots \underline{j_l} \dots i_n} \bra{j_1 \dots \underline{i_k} \dots j_n}\in \mathcal{M}_d^{\otimes n}\,.
\end{equation}
\end{definition}
In words, the $k$'th ``ket" exchanges site with the $l$'th ``bra" (see Fig.~\ref{fig:reshuffling_kl}).
We illustrate its use with a small example that can be graphically seen in Fig.~\ref{fig:example11}.

\begin{example}\label{example11}
Let $A=|\alpha\rangle\langle\beta|$, then
\begin{align}
    &\tr_1[(1234)^{T_4}A\ot \one\ot\one\ot\one]\nn\\
    &=\tr_1[(|lijl\rangle\langle ijkk|)A\ot\one\ot\one\ot\one]\nn\\
    &=|\alpha\rangle\langle k|\ot|k\rangle\langle l|\ot|\beta\rangle\langle l| \nn\\
    &=\left[(A\ot \one\ot \one)^{R_{3,2}}\right]^{R_{3,1}}\,.
\end{align}
\end{example}

\begin{figure*}[tbp]
\begin{center}
   \includegraphics[width=0.8\textwidth]{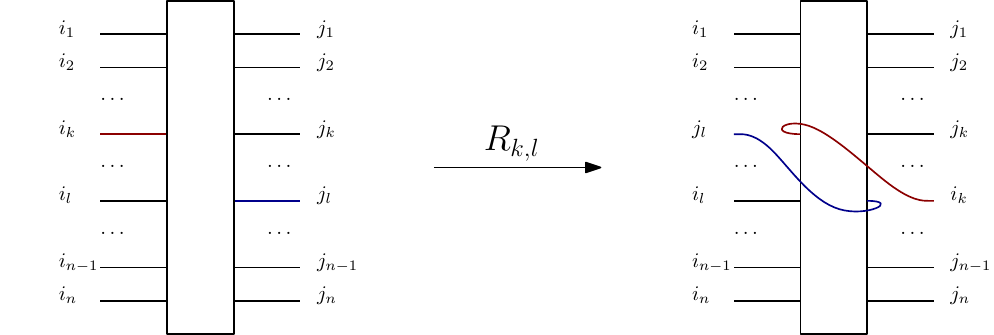}
\caption{Graphical representation of the reshuffling of sites $k,l$. The action of the reshuffling operation interchanges the $k$'th ``ket" leg with the $l$'th ``bra" leg.} 
\label{fig:reshuffling_kl}
\end{center}
\end{figure*}

\begin{figure}[tbp]
\begin{center}
   \includegraphics[width=\columnwidth]{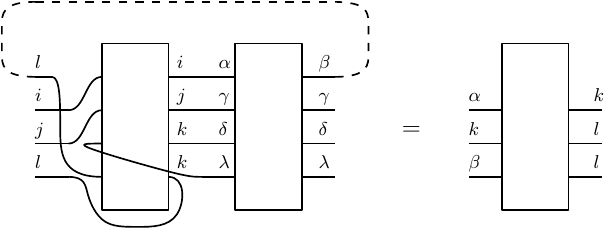}
\caption{Graphical representation of Example~\ref{example11}. The dashed line represents the partial trace.} 
\label{fig:example11}
\end{center}
\end{figure}

This generalizes in the following way.

\begin{widetext}
\begin{prop}\label{prop:1tok-1}
Let $A\in\mathcal{M}_d$ and consider the cycle $(1\dots k)$. Then
\begin{align}
    \tr_{1}[(1\dots k)^{T_k} A\ot \one\ot\dots\ot \one]=\left[\Bigg[\Big[[A\ot\one\ot \one\ot \cdots \ot \one]^{R_{k-1,k-2}}\Big]^{R_{k-1,k-3}}\Bigg]^{\dots}\right]^{R_{k-1,1}}\,,
\end{align}
where $R_{k,l}$ stands for reshuffling of sites $k,l$.
\end{prop}
The proof of this Proposition can be found in Appendix \ref{app:proofs}.
\end{widetext}

We can also write this result in terms of a single permutation. To do so, consider the following cyclic permutation with $k'=k-1$
\begin{align}
    \pi = ((2k'-1)(2k'-2)\dots(k'+1) k')\,,
    \label{permutationreshuffling}
\end{align}
and consider the mapping $\tau$ that switches the ``ket"-side of an operator into a ``bra", 
defined by
\begin{align}
    \tau(\ket{i_1 \dots i_n}\bra{j_1 \dots j_n})
    &= \ket{i_1 \dots i_n}\ket{j_1 \dots j_n} \,,\label{thetakettobra}
    \nn\\
    \tau^{-1}(\ket{i_1 \dots i_n}\ket{j_1 \dots j_n})
    &= \ket{i_1 \dots i_n}\bra{j_1 \dots j_n}\,.
\end{align}
It is clear that $\tau^{-1}\circ \tau = \id$. Then, a permutation $\pi \in S_{2n}$ acts on $\ket{i_1 \dots i_n}\bra{j_1 \dots j_n}$ by 
\begin{equation}
    \pi (\ket{i_1 \dots i_n}\bra{j_1 \dots j_n}) := 
    \tau^{-1} \pi \tau(\ket{i_1 \dots i_n}\bra{j_1 \dots j_n})\,.
\end{equation}
With these definitions, the concatenation of reshufflings can be written as a single permutation.
\begin{prop}\label{prop:1tok-1permutation}
Let $A\in\mathcal{M}_d$ and consider the cycle $(1\dots k)$. Then
\begin{align}
    &\tr_{1}[(1\dots k)^{T_k} A\ot \one\ot\dots\ot \one] \nn\\ &=[\tau^{-1}\pi\tau(A \ot\one\ot\dots\ot \one)]\,,
\end{align}
where $\pi$ and $\tau$ are defined in Eq.~\eqref{permutationreshuffling} and Eq.~\eqref{thetakettobra}
respectively.
\end{prop}
The proof of this follows the same idea as the previous one and can be found in Appendix~\ref{app:proofs}.

\section{Maps and matrix functions from irreps of $U^{\otimes (n-k)}\otimes \bar{U}^{\otimes k}$}\label{sec:irrepssection}
In this section, we construct completely positive multilinear maps that are invariant under the action of $U^{\ot (n-k)} \ot \bar{U}^{\ot k}$.
For this, we exploit the irreps of the algebra of partially transposed permutation operators $\mathcal{WBA}(n,k,d)$ over last $k$ subsystems. This is equivalent to analysing the irreps of the commutant of $U^{\otimes (n-k)}\otimes \bar{U}^{\otimes k}$ operations, where the bar denotes complex conjugation. Here, as an additional parameter, we fix a local dimension $d$, since we work on the representation space $(\mathbb{C}^d)^{\otimes n}$. However, to consider abstract irreps one does not have to use this parameter.  For the reader's convenience we refer to Appendix~\ref{appA:rep_th} for a short overview of the topic. It also contains the necessary references for further reading.

The irreducible representations of the symmetric group $S_n$ are labelled by so called Young diagram $\alpha \vdash n$, which are represented by a collection of $n$ boxes, arranged in left-justified rows, with rows lengths in increasing order.
Let $\alpha \vdash (n-2k)$, $\mu \vdash (n-k)$ be two irreducible representations of symmetric groups $S_{n-2k}$ and $S_{n-k}$ respectively with the restriction that $\mu$ can be obtained from $\alpha$ by adding $k$ boxes. With these irreducible representations we can associate projectors $P_{\alpha}$, $P_{\mu}$ on the respective isotypic components. 
Due to the results from Ref.~\cite{StuNJP,Mozrzymas2021optimalmultiport} the irreducible projectors of the algebra $\mathcal{WBA}(n,k,d)$, denoted as $F_{\mu}(\alpha)$, are related to irreducible projectors $P_{\alpha}$, $P_{\mu}$ in the following way\footnote{The notation slightly differs to~\cite{StuNJP,Mozrzymas2021optimalmultiport}, as we wanted to make clear that the $P_\alpha$ and $\sigma^{(k)}$ have distinct supports. Putting the tensor product achieves that.}:
\begin{equation}
\label{eq:FsingleT}
F_{\mu}(\alpha)=\frac{1}{\gamma_{\mu}(\alpha)}P_{\mu}\sum_{\eta \in \mathcal{S}_{n,k}}\eta^{-1}P_{\alpha}\otimes \sigma^{(k)} \eta\,,
\end{equation}
with the normalisation factor:
\begin{equation}
 \gamma_{\mu}(\alpha)=k!\binom{n-k}{k}\frac{m_{\mu}}{m_{\alpha}}\frac{d_{\alpha}}{d_{\mu}}\,.   
\end{equation}
The sum in~\eqref{eq:FsingleT} runs over the coset $\mathcal{S}_{n,k}= \frac{S(n-k)}{S(n-2k)}$ and  the $\mathcal{WBA}$ operator $\sigma^{(k)}$ is given by~\eqref{sigmagen1}. In the particular case $k=1$ we have $\sigma^{(1)}=(n-1,n)^{T_n}$.

The numbers $d_{\mu}$ and $d_{\alpha}$ are the dimensions, and $m_{\mu}$ and $m_{\alpha}$ are the multiplicities of the irreps in the Schur-Weyl duality labelled by $\mu$ and $\alpha$, respectively. 

Having a recipe to construct projectors $F_{\mu}(\alpha)$ one can construct new examples of multilinear maps by adapting Proposition~\ref{prop:general_pos_map}:
\begin{align}\label{eq:multilinearmapF}
    &\Lambda_{\alpha,\mu}(X_1,\dots,X_n) \\
    &=
    \tr_{1\dots n}
    \big[F_{\mu}(\alpha)(X_1\ot\dots\ot X_n\ot\one\ot\dots\ot\one)
    \big]\,. \nonumber
\end{align}

We will now show particular examples of this general case, by looking at the irreps of $U^{\ot 3}\ot \bar{U}$ and $U^{\ot 3}\ot \bar{U}^{\ot 2}$.

\begin{example}[Case $n=4,k=1,d=2$]
\ytableausetup{smalltableaux}
The possible choices for $\mu$ are $(3), (2,1)$ and $(1,1,1)$,
while $\a$ can be $(2)$ and $(1,1)$.
The allowed combinations of $\mu$ and $\alpha$ are those that differ by a single box. 
Additionally, 
since we work with $d=2$ we know that antisymmetric space, labelled by $(1,1,1)$, is not represented.
This gives three projectors corresponding to the partitions
\begin{align}
(\alpha_1,\mu_1) &= \Big(\ydiagram{2}, \ydiagram{3}\Big)\,,
\nn \allowdisplaybreaks\\
(\alpha_1,\mu_2) &= \Big(\ydiagram{2}, \ydiagram{2,1}\Big)\,, \nonumber\\
%
(\alpha_2,\mu_3) &= \Big(\ydiagram{1,1}, \ydiagram{2,1}\Big)\,.
\end{align}

For this example, we consider the projector $F_{\mu_2}(\alpha_1)$ generated from the symmetric partition $(2)$. 
In this case the normalisation constant is $\gamma_{\mu_2}(\alpha_1)=~1$.
The Young projectors $P_{\alpha_1}$ and $P_{\mu_2}$ are given by
\begin{align}
    P_{\alpha_1}&=\frac{1}{2}\left[\id+(12)\right], \nn\\
    P_{\mu_2}&=\frac{1}{3}\left[2\id-(123)-(132)\right]\,.
\end{align}
With this we arrive at
\begin{widetext}
\begin{align}
    F_{\mu_2}(\alpha_1)&=\frac{1}{3}\Big(2\id-(123)-(132)\Big)\Big(P_{\alpha_1} \sigma^{(1)}+(23)P_{\alpha_1} \sigma^{(1)}(23)+(13)P_{\alpha_1} \sigma^{(1)}(13)\Big)\nonumber\\
    &=\frac{1}{6}\Big\{-[(123)(14)^{T_4}+(123)(24)^{T_4}+(123)(34)^{T_4}+(132)(14)^{T_4}+(132)(24)^{T_4}+(132)(34)^{T_4}]\nonumber\\&+2[(14)^{T_4}+(24)^{T_4}+(34)^{T_4}]\Big\}\,,
\end{align}
\end{widetext}
where $\sigma^{(1)}=(34)^{T_4}$, and the second equality is obtained using the graphical notation in Appendix~\ref{app:graphicaltrick}. 
With Prop.~\ref{prop:general_pos_map},
the projector $ F_{\mu_2}(\alpha_1)$ 
gives the following positive map from 2 to 2 tensor factors,
\begin{align}\label{eq:2to2}
& \Lambda_{\alpha_1,\mu_2}(A,B) = 
  \tr_{12}[F_{\mu_2}(\alpha_1) A\ot B\ot \one\ot \one]\nn\\
    &= -\frac{1}{6}[(B\ot A^T)^R+(\one\ot A^TB^T)^R+(\one\ot BA)^R\nn\allowdisplaybreaks\\
    &\phantom{=}+(\one\ot B^TA^T)^R+(A\ot B^T)^R+(AB\ot \one)^R]\nn\\
    &\phantom{=}+\frac{1}{3}[\tr(B)(\one\ot A^T)+\tr(A)(\one\ot B^T)\nn\\
    &\phantom{=}+\tr(A)\tr(B)(\one\ot\one)^R]  \,,
\end{align}
and a map from $3$ to $1$ tensor factors,
\begin{align}\label{eq:3to1}
& \Lambda_{\alpha_1,\mu_2}(A,B,C) = 
   \tr_{123}[F_{\mu_2}(\alpha_1) A\ot B\ot C\ot \one]\nn\\
   &=-\frac{1}{6}[(ACB)^T+(BAC)^T+(CBA)^T+(ABC)^T\nn\\
   &\phantom{=}+(BCA)^T+(CAB)^T]+\frac{1}{3}[\tr(B)\tr(C)A^T\nn\\
   &\phantom{=}+\tr(A)\tr(C)B^T+\tr(A)\tr(B)C^T]   \,.
\end{align}
\end{example}

\begin{example}[Case $n=5,k=2,d=2$] 
The possible choices for $\nu$ are $(3),(2,1)$ and $(1,1,1)$, and the trivial choice $(1)$ for $\beta$. Again the antisymmetric space labelled by $(1,1,1)$ is not represented. Allowed combinations of $\nu$ and $\beta$ are those that differ by two boxes, namely
\begin{align}
&(\beta,\nu_1) = \Big(\ydiagram{1}, \ydiagram{3}\Big),
&
(\beta,\nu_2) = \Big(\ydiagram{1}, \ydiagram{2,1}\Big). 
\end{align}
Now consider the projector $F_{\nu_2}(\beta)$. In this case the normalisation constant is $\gamma_{\nu_1}(\alpha)=3$. Young projectors $P_{\beta}$ and $P_{\nu}$ are given by
\begin{align}
    P_{\beta}&=\id\,,\nn\\
    P_{\nu}&=\frac{1}{3}\left[2\id-(123)-(132)\right]\,.
\end{align}
With the recipe given in~\eqref{Fgen1} and~\eqref{sigmagen1} one has $\mathcal{S}_{5,2}=S(3)$. Noting that $\sigma^{(2)}=(25)^{T_5}(34)^{T_4}$, then 
\begin{align}
F_{\nu}(\beta)&=\frac{1}{9}\Bigg\{\Big[2\id-(123)-(132)\Big]\Big[\id\sigma^{(2)}\id\nonumber\\
&+(12)\sigma^{(2)}(12)+(23)\sigma^{(2)}(23)\nonumber\\
&+(13)\sigma^{(2)}(13)+(132)\sigma^{(2)}(123)\nonumber\\
&+(123)\sigma^{(2)}(132)\Big]\Bigg\}\,.
\end{align}

The projector $F_{\nu}(\beta)$ 
gives with 
Proposition~\ref{prop:general_pos_map}
the following positive map from 3 to 2 tensor factors,
\begin{align}
& \Lambda_{\beta,\nu}(A,B,C) = 
  \tr_{123}[F_{\nu}(\beta) A\ot B\ot C\ot \one\ot \one]\nn\\
  &=\frac{1}{9}\Big\{4\tr(A)\tr(B)(\one\ot C)^R+2\tr(A)\tr(C)(\one\ot B^T)^R\nonumber\\
  &\phantom{=}+2\tr(B)\tr(C)(\one\ot A)^R+(\one\ot BCA)^R(\one\ot ACB)^R\nonumber\allowdisplaybreaks\\
  &\phantom{=}-2(\one\ot CBA)^R-(\one\ot BAC)^R-2(\one\ot ABC)^R\nonumber\\
  &\phantom{=}-\tr(A)\tr(C)(\one\ot B)^R-3(\one\ot CAB)^R\nonumber\\
  &\phantom{=}-\tr(B)\tr(C)(\one\ot A)^R\Big\}\,.
\end{align}
\end{example}

Again the same operations appear: trace, transposition, and reshuffling. In general, it is possible to work with larger number of partial transposes and systems by following~\cite{Mozrzymas2021optimalmultiport}. 
However, increasing the number of system and partial transpositions likely requires 
a dedicated symbolic software.

\section{Conclusions}
We have introduced a formalism that connects the partial transpose and reshuffling operations with 
entanglement witnesses that have $U^{\otimes (n-k)}\otimes {\bar{U}}^{\otimes k}$-invariance.
This allows to construct entanglement witnesses starting from positive maps, and vice versa, 
to obtain positive maps from witnesses and block-positive operators. 
Because of the connection with the walled Brauer algebra, 
these maps allow for closed-form expressions involving 
partial trace, 
partial transpose, 
and reshuffling operations.

For further research, it would be interesting to understand whether immanant-type inequalities related to
the walled Brauer algebra can be developed in analogy
to Ref.~\cite{huber2021matrix}.
Also, inspired by Ref.~\cite{holmes2021nonlinear}, 
it would be interesting to understand the power of nonlinear quantum algorithms 
whose transformations are given by the walled Brauer algebra.

\onecolumngrid

\appendix
\section{Primer on representation theory of the symmetric group and the walled Brauer algebra}
\label{appA:rep_th}
In the first part of this Appendix, we briefly describe the representation theory of the symmetric group $S_n$ with the corresponding group algebra $\mathbb{C}[S_n]$. Next, we state the most basic facts about the irreducible representations of the algebra of partially transposed permutation operators. For the full presentation, we refer the reader to the manuscripts cited here and in the main text.

A partition $\alpha$ of a natural number $n$, denoted by $\alpha \vdash n$, is a sequence of positive numbers $\alpha=(\alpha_1,\alpha_2,\ldots,\alpha_r)$ such that
\begin{align}
\alpha_1\geq \alpha_2\geq \ldots \geq \alpha_r,\qquad \sum_{i=1}^r\alpha_i=n\,.
\end{align}
This can be represented graphically: every partition can be visualised as a \textit{Young diagram} --a collection of boxes arranged in left-justified rows. Later, when having two Young diagrams $\alpha \vdash n-1$ and $\mu \vdash n$, we will write $\mu \in \alpha$ to say that the diagram $\mu$ is obtained from the diagram $\alpha$ by adding a single box.

Now, let us consider a permutational representation of the symmetric group $S_n$ in the Hilbert space $\mathcal{H}=(\mathbb{C}^d)^{\otimes n}$. The elements of $S_n$ permute vectors in $(\mathbb{C}^d)^{\otimes n}$ according to a given permutation $\pi$:
\begin{align}
\forall \pi\in S_n \qquad \pi |v_1\>\otimes |v_2\>\otimes
	\ldots \otimes |v_n\>:=|v_{\pi ^{-1}(1)}\>\otimes |v_{\pi
			^{-1}(2)}\>\otimes \ldots \otimes |v_{\pi ^{-1}(n)}\>\,.
\end{align}
Formally, we should distinguish between an abstract permutations $\pi,\sigma,\ldots \in S_n$ and their representations. However, to compress the notation in the both cases we use the same symbols whenever it is clear from the context.

The permutation representation of $S_n$ extends to the
representation of the group algebra $\mathbb{C}[S_n]$,
\begin{align}
\label{CSn}
\mathbb{C}[S_n]:= \operatorname{span}_{\mathbb{C}}\{\pi :\pi \in S_n\}\subset \operatorname{Hom}(\mathcal{(\mathbb{C}}^{d})^{\otimes n})\,.
\end{align}

There is a one-to-one correspondence between the set of all Young diagrams and the irreducible representations of $S_n$. Namely, for a fixed number $n$, every Young diagram labels different irreducible representation of $S_n$. This means that the number of Young diagrams, for a given $n$, determines the number of nonequivalent irreps of $S_n$ in an abstract decomposition. Working in the representation space $(\mathbb{C}^{d})^{\otimes n}$, in every decomposition of $S_n$ into irreps, we take Young diagrams $\alpha$ whose height $\operatorname{ht}(\alpha)$ is at most $d$.

The diagonal action of the general linear group $GL_d(\mathbb{C})$ of invertible complex matrices and of the symmetric group on $(\mathbb{C}^d)^{\otimes n}$ commute, 
\begin{align}
\pi(X\otimes \ldots \otimes X)=(X\otimes \ldots \otimes X)\pi,
\end{align}
where $\pi \in S_n$ and $X\in GL_d(\mathbb{C})$. Due to the above relation, there exists a basis in which the tensor product space $(\mathbb{C}^d)^{\otimes n}$ can be decomposed as
\begin{align}
\label{wedd}
(\mathbb{C}^d)^{\otimes n}=\bigoplus_{\substack{\alpha \vdash n \\ \operatorname{ht}(\alpha)\leq d}} \mathcal{U}_{\alpha}\otimes \mathcal{S}_{\alpha}\,,
\end{align}
where the symmetric group $S_n$ acts on the representation space $\mathcal{S}_{\alpha}$, and the general linear group $GL_d(\mathbb{C})$ acts on the representation space $\mathcal{U}_{\alpha}$, labelled by the same partitions. The unitary group $U(d)$  satisfies the same kind of decomposition, with the same set of irreps.

From the decomposition given in expression~\eqref{wedd} we deduce that for a given irrep $\alpha$ of $S_n$, the space $\mathcal{U}_{\alpha}$ is a multiplicity space of dimension $m_{\alpha}$ (multiplicity of irrep $\alpha$), while the space $\mathcal{S}_{\alpha}$ is a representation space of dimension $d_{\alpha}$ (dimension of irrep $\alpha$).
Finally, with every subspace $\mathcal{U}_{\alpha}\otimes \mathcal{S}_{\alpha}$ we associate a \textit{Young projector}
\begin{align}
\label{Yng_proj}
P_{\alpha}=\frac{d_{\alpha}}{n!}\sum_{\pi \in S_n}\chi^{\alpha}(\pi^{-1})\pi,\quad \text{with}\quad \tr P_{\alpha}=m_{\alpha}d_{\alpha},
\end{align}
where $\chi^{\alpha}(\pi^{-1})$ is the irreducible character associated with the irrep indexed by $\alpha$. The symbols $m_{\alpha}, d_{\alpha}$ denote, respectively, the multiplicity and the dimension of an irrep $\alpha$.

Having the definition of the algebra $\mathbb{C}[S_n]$ [Eq.~\eqref{CSn}], we naturally extend it to the algebra of partially transposed permutation operators (walled Brauer algebra) with respect to last subsystem
\begin{align}
	\mathcal{WBA}(n,1,d):= \operatorname{span}_{\mathbb{C}}\{\pi^{T_n}:\sigma \in S_n\}\subset \operatorname{Hom}(\mathcal{(\mathbb{C}}^{d})^{\otimes n}),
\end{align}
where $T_n$ denotes partial transposition with respect to $n-$th system.
The elements of $\mathcal{WBA}(n,1,d)$ commute with the deformed action of unitary group $U(d)$ of the form $U^{\otimes (n-1)}\otimes \bar{U}$, where the bar denotes complex conjugation. Moreover, there exists a basis in which we have a decomposition of $(\mathbb{C}^{d})^{\otimes n}$  similar as it is in Eq.~\eqref{wedd}. Namely, we have
\begin{align}
\label{wedd2}
(\mathbb{C}^d)^{\otimes n}=\bigoplus_{b} \mathcal{U}_{b}\otimes \mathcal{S}_{b}^{\mathcal{WBA}}\,,
\end{align}
where the index $b$ labels all irreducible representations of the considered algebra $\mathcal{WBA}(n,1,d)$. In fact, every partially transposed permutation operator is represented non-trivially on $\mathcal{S}_{a}^{WBA}$, and it gives an irreducible matrix representation of the walled Brauer algebra. The algebra $\mathcal{WBA}(n,1,d)$ is a direct sum of two ideals,  
\begin{equation}
\label{A_decomp}
\mathcal{WBA}(n,1,d)=\mathcal{M}\oplus \mathcal{S}=F \ \mathcal{WBA}(n,1,d) \  F\oplus
(id_{\mathcal{WBA}}-F)\mathcal{WBA}(n,1,d)(id_{\mathcal{WBA}}-F),
\end{equation}
where the idempotent $F=\sum_{\alpha \vdash n-2}\sum_{\mu \in \alpha}F_{\mu }(\alpha )$ is the identity on the ideal $\mathcal{M}$, i.e., $F=id_{\mathcal{M}}$, and $\id_{\mathcal{WBA}}$ is the identity operator on the whole space. The operators $F_{\mu }(\alpha )$ are projectors onto the irreps of $\mathcal{WBA}(n,1,d)$ contained in the ideal  $\mathcal{M}$. This means that the index $a$ in Eq.~\eqref{wedd2} is in fact a pair $b=(\alpha,\mu)$, such that $\mu \in \alpha$ for the irreps contained in $\mathcal{M}$. The explicit form of the projectors $F_{\mu }(\alpha )$, due to results from Ref.~\cite{StuNJP}, is
\begin{align}
\label{explicit}
F_{\mu}(\alpha)=\frac{1}{\gamma_{\mu}(\alpha)}P_{\mu}\sum_{a=1}^{n-1}(a,n-1)P_{\alpha}\otimes (n-1,n)^{T_n}(a,n-1)\,,
\end{align}
where $P_{\alpha},P_{\mu}$ are the Young projectors onto irreducible spaces labelled by the Young diagrams $\alpha \vdash n-2$ and $\mu \vdash n-1$, respectively, defined in~\eqref{Yng_proj}, and $\gamma_{\mu}(\alpha)=(n-1)\frac{m_{\mu}d_{\alpha}}{m_{\alpha}d_{\mu}}$.
The described algebra can be also studied for larger number of partial transposes, i.e. $\mathcal{WBA}(n,k,d)$, where $k>1$. In the general case, the operator $F_{\mu}(\alpha)$ has a similar, however more complicated form to those from~\eqref{explicit}:
\begin{equation}
\label{Fgen1}
F_{\mu}(\alpha)=\frac{1}{\gamma_{\mu}(\alpha)}P_{\mu}\sum_{\eta \in \mathcal{S}_{n,k}}\eta^{-1}P_{\alpha}\otimes \sigma^{(k)} \eta,\quad \text{with}\quad \gamma_{\mu}(\alpha)=k!\binom{n-k}{k}\frac{m_{\mu}}{m_{\alpha}}\frac{d_{\alpha}}{d_{\mu}},
\end{equation}
where the sum runs over the coset $\mathcal{S}_{n,k}= \frac{S(n-k)}{S(n-2k)}$ and $\alpha \vdash (n-2k)$, $\mu \vdash (n-k)$ such that $\mu\in\alpha$. The latter means that $\mu$ can be obtained from $\alpha$ by adding $k$ boxes. Moreover, by writing $\sigma^{(k)}$ we understand the following operator
\begin{equation}
\label{sigmagen1}
\sigma^{(k)}:=(n-2k+1,n)^{T_n}(n-2k+2,n-1)^{T_{n-1}}\cdots (n-k,n-k+1)^{T_{n-k+1}},
\end{equation}
where $T_n,T_{n-1},\ldots,T_{n-k+1}$ denote partial transposition with respect to particular subsystems. Note that inserting $k=1$ we get $\mathcal{S}_{n,1}=S(n-1)/S(n-2)$, which is composed from the permutations of the form $(a,n-1)$, for $1\leq a \leq n-1$. Moreover, we have in this case $k!\binom{n-k}{k}=1!\binom{n-1}{1}=n-1$ and $\sigma^{(1)}=(n-1,n)^{T_n}$ as it should be. 
For more details see for example Ref.~\cite{Mozrzymas2021optimalmultiport}, where this algebra has been studied in the context of quantum teleportation protocols.

\section{Eggeling-Werner maps}\label{app:wernermaps}
Following Ref.~\cite{Eggeling2001}, we recall the entanglement properties of tripartite Werner states and obtain associated positive maps.
Moreover, we list all the maps that by performing partial transposes over different subsystems in Eq.~\eqref{Wernermaps}.

Every tripartite $U \ot U \ot U$-invariant state $\varrho$ (also known as Werner state) is uniquely characterised by $(r_+,r_-,r_0,r_1,r_2)\in\mathbb{R}^6$ satisfying
\begin{align}
    r_+,r_-,r_0\geq 0, \quad r_++r_-+r_0=1 \quad \text{and}\quad r_1^2+r_2^2+r_3^2\leq r_0^2\,.
\end{align}
where 
\begin{align}
    r_+=\frac{d}{6}(d^2+3d+2)c_+,\qquad r_-=\frac{d}{6}(d^2-3d+2)c_-\,, \qquad
    r_i=\frac{2d}{3}(d^2-1)c_i,\quad i=0,1,2,3\,.
    \label{cintermsofr}
\end{align}
Each $\varrho$ can thus be written in the following decomposition
\begin{align}
    \varrho=\sum_{k=+,-,0,1,2,3} c_k R_k\,.
    \label{rhofromRk}
\end{align}
Relating this to the decomposition of $\varrho$ in terms of permutation operators,
\begin{eqnarray}
    \varrho=\alpha_1\id+\alpha_2(12)+\alpha_3(23)+\alpha_4(31)+\alpha_5(123)+\alpha_6 (321)\,,
    \label{rhointermsofalphas}
\end{eqnarray}
where the coefficients $\alpha_k$ are defined as
\begin{align}\label{alphas_i}
    \alpha_1&=\frac{c_+}{6}+\frac{c_-}{6}+\frac{2c_0}{3}\,, 
    &\alpha_2&=\frac{c_+}{6}-\frac{c_-}{6}-\frac{c_2}{\sqrt{3}}\,,
    &\alpha_3&=\frac{c_+}{6}-\frac{c_-}{6}+\frac{2c_1}{3},\nonumber\\ 
    \alpha_4&=\frac{c_+}{6}-\frac{c_-}{6}-\frac{c_1}{3}-\frac{c_2}{\sqrt{3}}\,,
    &\alpha_5&=\frac{c_+}{6}+\frac{c_-}{6}-\frac{c_0}{3}+\frac{ic_3}{\sqrt{3}}\,,
    &\alpha_6&=\frac{c_+}{6}+\frac{c_-}{6}-\frac{c_0}{3}-\frac{ic_3}{\sqrt{3}}\,,
\end{align}
and the coefficients $c_k$ for $k=+,-,0,1,2,3$ are defined in Eq.~\eqref{cintermsofr}.

Now for different values of $j$ we consider the following maps
\begin{align*}
    f_S(A)   &=\tr_1[\varrho^{T_s}A\ot\one\ot\one] \,, \\
    g_S(A,B) &=\tr_{12}[\varrho^{T_S}A\ot B\ot\one]\,.
\end{align*}

The results for all the possible subsets $S$ are given in Table~\ref{tablemaps}. Note that the map $f_3(A^T)$ is equivalent to the map $f_{13}(A)$, and similarly $f_2(A^T)=f_{12}(A)$. For the $g$ maps a similar thing happens, i.e. $g_{13}(A^T,B^T)=g_{23}(A,B)$ and $g_1(A^T,B^T)=g_2(A,B)$.

\begin{table}[tbp]
\centering
\begin{tabular}{@{}ll@{}}
map  &   evaluation\\ 
\midrule
$f_1$  & $\alpha_1\tr(A)(\one\ot\one)+\alpha_2(A^T\ot\one)+\alpha_3\tr(A)((\one\ot\one)^R)^{T_2}+\alpha_4(\one\ot A^T)+\alpha_5((A^T\ot\one)^R)^{T_2}+\alpha_6((\one\ot A)^R)^{T_2}$ \\ 
$f_2$  &  $\alpha_1\tr(A)(\one\ot\one)+\alpha_2(A^T\ot\one)+\alpha_3\tr(A)(\one\ot\one)^R+\alpha_4(\one\ot A)+\alpha_5(\one\ot A)^R+\alpha_6(A^T\ot\one)^R$\\ 
$f_3$  & $\alpha_1\tr(A)(\one\ot\one)+\alpha_2(A\ot\one)+\alpha_3\tr(A)(\one\ot\one)^R+\alpha_4(\one\ot A^T)+\alpha_5(A\ot\one)^R+\alpha_6(\one\ot A^T)^R$ \\ 
$f_{12}$ & $\alpha_1\tr(A)(\one\ot\one)+\alpha_2(A\ot\one)+\alpha_3\tr(A)(\one\ot\one)^R+\alpha_4(\one\ot A^T)+\alpha_5(\one\ot A^T)^R+\alpha_6(A\ot\one)^R$ \\ 
$f_{13}$ & $\alpha_1\tr(A)(\one\ot\one)+\alpha_2(A^T\ot \one)+\alpha_3\tr(A)(\one\ot\one)^R+\alpha_4(\one\ot A)+\alpha_5(A^T\ot\one)^R+\alpha_6(\one\ot A)^R$ \\ 
$f_{23}$ & $\alpha_1\tr(A)(\one\ot\one)+\alpha_2(A^T\ot\one)+\alpha_3\tr(A)((\one\ot\one)^R)^{T_2}+\alpha_4(\one\ot A^T)+\alpha_5((\one\ot A)^R)^{T_2}+\alpha_6((A^T\ot\one)^R)^{T_2}$ \\
\midrule
$g_1$  & $\alpha_1\tr(A)\tr(B)\one+\alpha_2\tr(A^TB)\one+\alpha_3\tr(A)B+\alpha_4\tr(B) A^T+\alpha_5BA^T+\alpha_6A^TB$ \\
$g_2$  &  $\alpha_1\tr(A)\tr(B)\one+\alpha_2\tr(AB^T)\one+\alpha_3\tr(A)B^T+\alpha_4\tr(B) A+\alpha_5B^TA+\alpha_6AB^T$\\
$g_3$  & $\alpha_1\tr(A)\tr(B)\one+\alpha_2\tr(AB)\one+\alpha_3\tr(A)B^T+\alpha_4\tr(B) A^T+\alpha_5A^TB^T+\alpha_6B^TA^T$ \\ 
$g_{12}$ & $\alpha_1\tr(A)\tr(B)\one+\alpha_2\tr(A^TB^T)\one+\alpha_3\tr(A)B^T+\alpha_4\tr(B) A^T+\alpha_5B^TA^T+\alpha_6A^TB^T$ \\ 
$g_{13}$ & $\alpha_1\tr(A)\tr(B)\one+\alpha_2\tr(A^TB)\one+\alpha_3\tr(A)B^T+\alpha_4\tr(B)A+\alpha_5AB^T+\alpha_6B^TA$ \\ 
$g_{23}$ & $\alpha_1\tr(A)\tr(B)\one+\alpha_2\tr(A^TB)\one+\alpha_3\tr(A)B+\alpha_4\tr(B) A^T+\alpha_5A^TB+\alpha_6BA^T$ \\ 
\bottomrule
\end{tabular}
\caption{The table list all the maps of the form Eqs.~\eqref{Wernermaps} and \eqref{Wernermaps2}. The maps are positive if and only if $\rho^{T_S}$ is block-positive with respect to the partition $1|23$ for the $f_S$, and block-positive for the partition $1|2|3$ for the $g_S$'s (Proposition~\ref{wernermapspositive}).}
\label{tablemaps}
\end{table}

\section{Technical proofs}\label{app:proofs}
In this appendix, we will give the technical proofs of the Propositions \ref{prop:k-1to1}, \ref{prop:1tok-1}, and \ref{prop:1tok-1permutation}. We state them again for clarity. Let us start with Prop.~\ref{prop:k-1to1}.
\begin{prop*}[\ref{prop:k-1to1}]
Let $X_1,\dots,X_k\in \mathcal{M}_d$ and $T_j$ denote partial transpose on $j$. Consider the cycle $(k\dots 1)=(1\dots k)^{-1}$. Then
\begin{align}
    \tr_{1\dots k\backslash k}\big[(k\dots 1)^{T_j} X_1\ot\dots\ot X_k\big]=
        \begin{cases}
            X_1\cdots X_j^T\dots X_{k-1} X_k &\text{for } j\neq k \\
            (X_{1}\cdots X_{k-1})^T X_k &\text{for } j=k
        \end{cases}
\end{align}
\begin{align}
    \tr_{1\dots k\backslash 1}\big[(1\dots k)^{T_j} X_1\ot\dots\ot X_k\big]=
        \begin{cases}
             X_kX_{k-1}\cdots X_j^T \dots X_1  &\text{for } j\neq 1 \\
              X_k^TX_{k-1}\cdots X_2 X_1  &\text{for } j=1
        \end{cases}
\end{align}
\end{prop*}
\begin{proof}
Let $\{|\alpha\rangle\}$ be an orthonormal basis for $\mathbb{C}^d$. 
Decompose 
$X_i=\sum_{\alpha_i,\beta_i=1}^d \chi_{\alpha_i\beta_i}|\alpha_i\rangle\langle\beta_i|$ 
and consider that 
$(k\dots 1)^{T_j}$ 
can be written as
\begin{equation}
(|i_1 i_2\dots i_j\dots i_k\rangle\langle i_ki_1\dots i_{j-1}\dots i_{k-1}|)^{T_j}=|i_1 i_2\dots i_{j-1}i_{j-1} i_{j+1}\dots i_k\rangle\langle i_k i_1\dots i_{j-2}i_j i_{j}\dots i_{k-1}| \,.
\end{equation} 
Then, for $j\neq k$
\begin{align}
    &\quad\,\tr_{1\dots k\backslash k}\big[(k\dots 1)^{T_j} X_1\ot\dots\ot X_k\big]
    \nonumber\\\nonumber
    &=
    \tr_{1\dots k\backslash k}
    \big[
    (k\dots 1)^{T_j}
    \sum_{\substack{\alpha_1,\dots, \alpha_k, \\ 
                    \beta_1,\dots, \beta_k=1}
         }^d 
    \chi_{\alpha_1\beta_1}
    \cdots
    \chi_{\alpha_k\beta_k}
    |\alpha_1\rangle\langle\beta_1|
    \ot\dots\ot
    |\alpha_j\rangle\langle\beta_j|
    \ot\dots\ot
    |\alpha_k\rangle\langle\beta_k|
    \big]
    \\\nonumber
    &=\tr_{1\dots k\backslash k}
    \big[
    \sum_{\substack{\alpha_1,\dots, \alpha_k, \\ 
                    \beta_1,\dots, \beta_k=1}
         }^d 
    \chi_{\alpha_1\beta_1}
    \cdots\chi_{\alpha_k\beta_k}
    |i_1\rangle\langle i_k|\alpha_1\rangle\langle\beta_1|
    \ot\dots\ot
    |i_{j-1}\rangle\langle i_j|\alpha_j\rangle\langle\beta_j|
    \ot\dots\ot
    |i_k\rangle\langle i_{k-1}|\alpha_k\rangle\langle\beta_k|
    \big]   
    \nonumber\\
    &=\chi_{\beta_{1}\alpha_{2}}\cdots\chi_{\beta_{j}\beta_{j-1}}\chi_{\alpha_{j}\alpha_{j+1}}\cdots\chi_{\beta_{k-1}\alpha_{k}}|\alpha_{1}\rangle\langle\beta_{k}|\\\nonumber
    &=X_{1}\cdots X_j^T\cdots X_{k-1} X_k\,.
\end{align}
The other relations can be obtained in a similar way.
\end{proof}

Let us now move to the $1$ to $k-1$ tensor factor maps. Consider Prop.~\ref{prop:1tok-1}.
\begin{prop*}[\ref{prop:1tok-1}]
Let $A\in\mathcal{M}_d$, $(1\dots k)$ be a cyclic permutation and $T_k$ be the partial transpose on site $k$. Then
\begin{align}
    \tr_{1}\big[(1\dots k)^{T_k} A\ot \one\ot\dots\ot \one\big]=\left[\Big[\big[[A\ot\one\ot \one\ot \dots \ot \one]^{R_{k-1,k-2}}\big]^{R_{k-1,k-3}}\Big]^{\dots}\right]^{R_{k-1,1}}\,,
\end{align}
where $R_{k,l}$ stands for reshuffling of sites $k,l$.
\end{prop*}

\begin{proof}
Let $A=|\alpha\rangle\langle \beta|$ and let $(1\dots k)=|j_1 j_2\dots j_{k-1}j_k\rangle\langle j_2j_3\dots j_kj_{1}|$. Then the left hand side reads:
\begin{align}
    \text{LHS}&=\tr_{1}
    \big[(|j_1 j_2\dots k_{k-1}j_{1}\rangle\langle j_2j_3\dots j_{k}j_k|)
    \,\cdot\,
    ( |\alpha\rangle\langle\beta|\ot |i_2\rangle\langle i_2|\ot |i_3\rangle\langle i_3|
    \ot\dots\ot
    |i_k\rangle\langle i_k|)\big]
    \nonumber\\
    &=
    \tr_{1}\big[|j_1\rangle\langle j_2|\alpha\rangle\langle\beta|\ot|j_2\rangle\langle j_3|i_2\rangle\langle i_2|\ot\dots\ot
    |j_{k-1}\rangle\langle j_{k}|i_{k-1}\rangle\langle i_{k-1}|\ot|j_{1}\rangle\langle j_{k}|i_{k}\rangle\langle i_{k}|\big]
    \nonumber\\
    &=
    \delta_{j_2 \alpha}\delta_{j_1 \beta}\delta_{j_3 i_2}\cdots\delta_{j_{k} i_{k-1}}\delta_{j_k i_k}|j_2\rangle\langle i_2|\ot |j_3\rangle\langle i_3|
    \ot\dots\ot
    |j_{k-1}\rangle\langle i_{k-1}|\ot|j_{1}\rangle\langle i_k|\nonumber\\
    &=
    |\alpha\rangle\langle j_3|\ot |j_3\rangle\langle j_4|\ot\dots\ot |j_{k-1}\rangle\langle j_{k}|\ot|\beta\rangle\langle j_k|\,.
\end{align}
The right hand side reads
\begin{align}
    \text{RHS}&=
    \Big\{
    \big[
    (
    |\alpha\rangle\langle \beta|\ot |i_2\rangle\langle i_2|\ot\cdots\ot |i_{k-2}\rangle\langle i_{k-2}|\ot|i_{k-1}\rangle\langle i_{k-1}|
    )^{R_{k-1,k-2}}
    \big]{}^{R_{k-1,k-3}}{}^{\dots}
    \Big\}{}^{R_{k-1,1}}\nn\\
    &=
    \Big\{
    \big[
    |\alpha\rangle\langle \beta|\ot |i_2\rangle\langle i_2|
    \ot\dots\ot 
    |i_{k-2}\rangle\langle i_{k-1}|\ot|i_{k-2}\rangle\langle i_{k-1}|\big]^{R_{k-1,k-3}}{}^{\dots}\Big\}{}^{R_{k-1,1}}\nn\\
    &=\big[|\alpha\rangle\langle \beta|\ot |i_2\rangle\langle i_3|
    \ot\dots\ot 
    |i_{k-2}\rangle\langle i_{k-1}|\ot|i_{2}\rangle\langle i_{k-1}|\big]^{R_{k-1,1}}\nn\\
    &=|\alpha \rangle\langle i_2|\ot |i_2\rangle\langle i_3|
    \ot\dots\ot
    |i_{k-2}\rangle\langle i_{k-1}|\ot|\beta\rangle\langle i_{k-1}|\,.
\end{align}
Identifying $j_m$ with $i_{m-1}$ we see that the left hand side equals the right hand side. 
This ends the proof. 
\end{proof}

Finally, we consider the $1$ to $k-1$ tensor factors map written with a single permutation, i.e., Proposition \ref{prop:1tok-1permutation}.
\begin{prop*}[\ref{prop:1tok-1permutation}]
Let $A\in\mathcal{M}_d$, $(1\dots k)$ be a cyclic permutation and $T_k$ be the partial transpose on site $k$. Then
\begin{align}
    \tr_{1}[(1\dots k)^{T_k} A\ot \one\ot\dots\ot \one]=[\tau^{-1}\pi\tau(A \ot\one\ot\dots\ot \one)],
\end{align}
where $\pi$ is defined in Eq.~\eqref{permutationreshuffling} and $\theta$ is defined in Eq.~\eqref{thetakettobra}.
\end{prop*}

\begin{proof}
Let $(1\dots k)=|j_1 j_2\dots j_{k-1}j_k\rangle\langle j_2j_3\dots j_kj_{1}|$ and let $A=|\alpha\rangle\langle \beta|$. Then, the left hand side reads
\begin{align}
    \text{LHS}&=\tr_{1} \big[
    (|j_1 j_2\dots k_{k-1}j_{1}\rangle\langle j_2j_3\dots j_{k}j_k|) \, \cdot\, ( |\alpha\rangle\langle\beta|\ot |i_2\rangle\langle i_2|\ot |i_3\rangle\langle i_3|
    \ot\dots \ot 
    |i_k\rangle\langle i_k|) \big]\nonumber\\
    &=\tr_{1}[|j_1\rangle\langle j_2|\alpha\rangle\langle\beta|\ot|j_2\rangle\langle j_3|i_2\rangle\langle i_2|
    \ot\dots \ot
    |j_{k-1}\rangle\langle j_{k}|i_{k-1}\rangle\langle i_{k-1}|\ot|j_{1}\rangle\langle j_{k}|i_{k}\rangle\langle i_{k}|]\nonumber\\
    &=\delta_{j_2 \alpha}\delta_{j_1 \beta}\delta_{j_3 i_2}\cdots\delta_{j_{k} i_{k-1}}\delta_{j_k i_k}|j_2\rangle\langle i_2|\ot |j_3\rangle\langle i_3|
    \ot\dots\ot 
    |j_{k-1}\rangle\langle i_{k-1}|\ot|j_{1}\rangle\langle i_k|\nonumber\\
    &=|\alpha\rangle\langle j_3|\ot |j_3\rangle\langle j_4|\ot\cdots\ot |j_{k-1}\rangle\langle j_{k}|\ot|\beta\rangle\langle j_k|\,.
\end{align}
The right hand side reads
\begin{align}
    \text{RHS}&=\tau^{-1}\pi\tau(|\alpha i_2 i_3\dots i_{k-1}\rangle\langle \beta i_2 i_3\dots i_{k-1}|)=\tau^{-1}\pi(|\alpha i_2 i_3\dots i_{k-1}\rangle|\beta i_2 i_3\dots i_{k-1}\rangle)\nonumber\\
    &=\tau^{-1}(|\alpha i_2 i_3\dots \beta\rangle| i_2 i_3\dots i_{k-1}i_{k-1}\rangle)=|\alpha i_2 i_3\dots \beta\rangle\langle i_2 i_3\dots i_{k-1}i_{k-1}|)\nonumber\\
    &=|\alpha\rangle\langle i_2|\ot |i_2\rangle\langle i_3|\ot\cdots\ot |i_{k-2}\rangle\langle i_{k-1}|\ot|\beta\rangle\langle i_{k-1}|\,.
\end{align}
Identifying $j_k$ with $i_k$ we see that the left hand side equals the right hand side. 
This ends the proof. 
\end{proof}

\section{Graphical notation for permutations and partial transposes}\label{app:graphicaltrick}
In this Appendix, we represent the partial trace over a subsystem , the partial transposition and the reshuffling operation acting on a tensor and their products graphically.
This is similar to the diagrams of Refs.~\cite{Collins2010,Seilinger2011,Penrose1971}.

\begin{figure}[tbp]
  \centering
  \includegraphics[width=0.5\columnwidth]{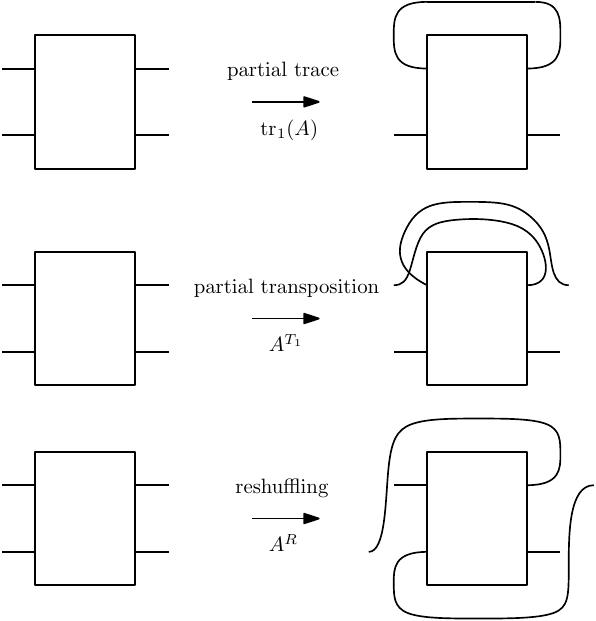}
  \caption{Graphical representations for the action of the partial trace (top), the partial transpose (middle), and reshuffling (bottom) on an operator $A$.}
  \label{fig:partialtransposition.pdf}
\end{figure}

Moreover, we explain a graphical notation to find the operator corresponding to concatenations of permutations and partial transposes. The first thing that one need to know is how to graphically represent single permutations and partial transpose. The idea is given in Fig.~\ref{fig:basic_graphictrick.pdf}.

If one then wants to combine a permutation and a partial transpose, one needs to first interchange the edges, and then the ket and the bra, as it can be seen in Fig.~\ref{fig:permutation_plus_PT.pdf}.
Further, if one wants to concatenate several of these permutations, partial transposes, and the combination of the two of them, one needs to arrange their graphical representation vertically and follow the lines, as it can be seen in Fig.~\ref{fig:allcombined_graphictrick.pdf}. Recall that the vertical ordering goes from the right to the left. In order to obtain the operator corresponding to this concatenation, one needs to set the same index to the connected edges, while recalling that the lower dots correspond to the ket and the upper ones to the bra.

\begin{figure}[tbp]
  \centering
  \includegraphics[width=0.85\columnwidth]{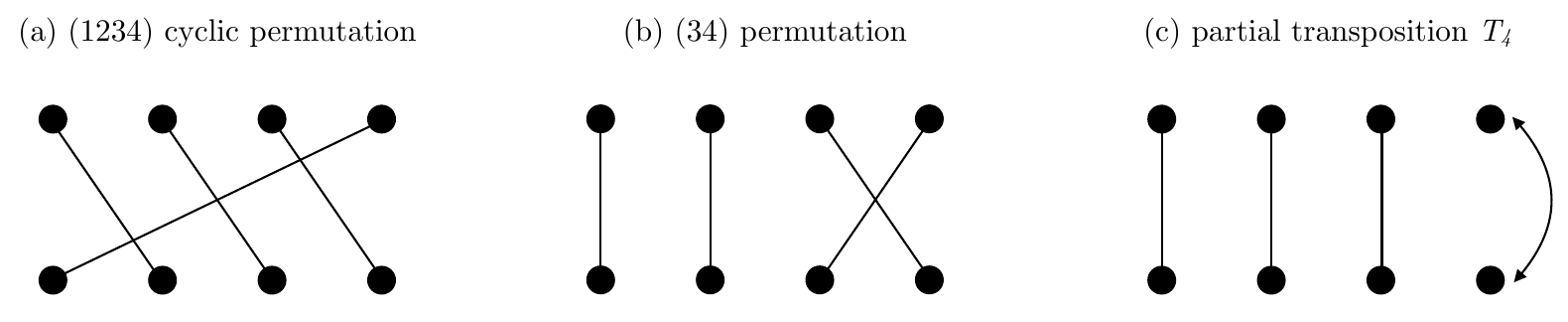}
  \caption{In (a) and (b) it can be seen that permutations are drawn as edges connecting the different parties (dots). In the cyclic permutation $(1234)$ the line of the first particle goes to the second, the second to the third and so on. In the $(34)$ permutation, the lines of the third and fourth particles are swapped. Clearly, elements (a) and (b) represent graphically elements of group algebra of $\mathbb{C}[S_4]$. In (c) it can be seen that a partial transpose consists of interchanging the ket and the bra, i.e., the upper dot (bra) with the lower dot (ket).}
  \label{fig:basic_graphictrick.pdf}
\end{figure}

\begin{figure}[tbp]
  \centering
  \includegraphics[width=0.53\columnwidth]{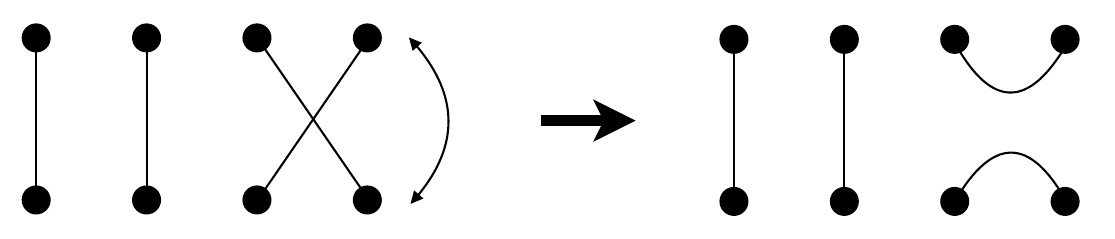}
  \caption{Combinations of a permutation and a partial transpose: $(34)^{T_4}$. The idea graphically is to first interchange the lines (do the permutation), and then the ket and the bra (do the partial transpose). The graphic on the right-hand side graphically represents an element of $\mathcal{WBA}(4,d,k)$. Please notice that in the pictorial representation we do not assign any specific dimension $d$, such diagrams are dimension independent.}
  \label{fig:permutation_plus_PT.pdf}
\end{figure}
\begin{figure}[tbp]
  \centering
  \includegraphics[width=0.85\columnwidth]{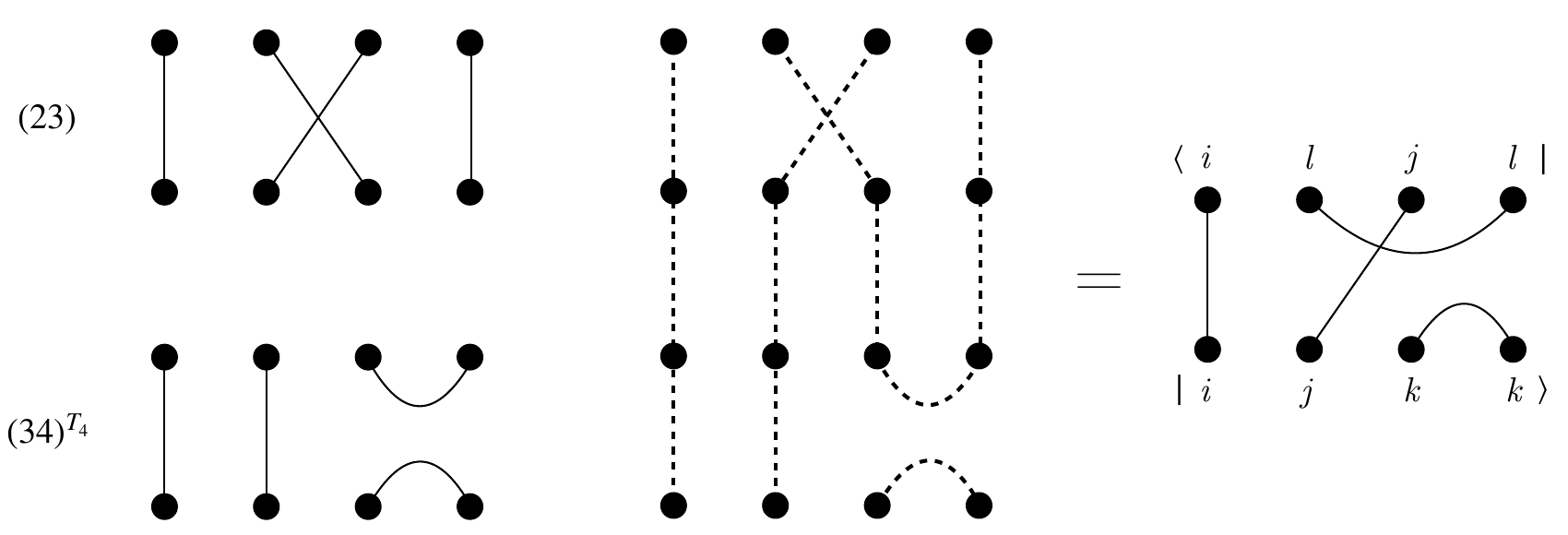}
  \caption{To obtain the operator of the $(34)^{T_4}(23)$, the trick is to follow the lines (from top to bottom). Finally, the connected edges correspond to the same index: the lower dots correspond to the ket and the upper ones to the bra in the sense that $(34)^{T_4}(23)=|ijkk\rangle\langle ilkl|$. The vertical ordering is first $(23)$ and then $(34)^{T_4}$. If one has more permutations to concatenate, the trick is the same, with more diagrams to connect vertically.}
  \label{fig:allcombined_graphictrick.pdf}
\end{figure}

\clearpage
\end{document}